\documentclass[10pt,onecolumn, draftclsnofoot]{IEEEtran}

\usepackage{graphics} % for pdf, bitmapped graphics files
\usepackage{epsfig} % for postscript graphics files
\usepackage{amsthm} % assumes amsmath package installed
\usepackage{amsmath} % assumes amsmath package installed
\usepackage{amssymb}  % assumes amsmath package installed
\usepackage{amsfonts}
\usepackage{color}
\usepackage[latin1]{inputenc}
\usepackage{psfrag}
\usepackage{booktabs}
\usepackage{mathbbol}
\usepackage{balance}
\usepackage{hyperref}
\usepackage{xcolor}% provides \colorlet
\usepackage{soul}% provides \hl
\usepackage[marginclue,inline]{fixme}
\fxsetup{
    status=draft,
    theme=color
}
\FXRegisterAuthor{ag}{aag}{\color{red}AG}
\FXRegisterAuthor{agt}{aagt}{\color{blue}AGT}

\definecolor{fxnote}{rgb}{0.8000,0.0000,0.0000}
\definecolor{agnote}{rgb}{0.0000,0.8000,0.0000}

\colorlet{fxnotebg}{yellow}

\makeatletter
\renewcommand*\FXTargetLayoutColor[2]{%
  \@fxuseface {target}%
  \hl{#2}}%
\makeatother
\newtheorem{theorem}{Theorem}

\newtheorem{remark}{Remark}
\newtheorem{lemma}{Lemma}

\newtheorem{corollary}{Corollary}

\newcommand{\1}{\mathbb{1}}

\newcommand{\ai}{a_\infty}

\newcommand{\aggiungi}[1]{\textcolor{black}{#1}} % testo aggiunto

\begin{document}

%\title{Asymptotic behaviors of threshold models in social networks}
\title{\LARGE \bf Asymptotic analysis of threshold models for social networks}

\author{
Andrea Garulli, Antonio Giannitrapani 
\thanks{The authors are with the Dipartimento di Ingegneria dell'Informazione e Scienze Matematiche, Universit\`{a} di Siena, Via Roma 56, 53100 Siena, Italy. \tt \{garulli,giannitrapani\}@dii.unisi.it}
}

\date{}

\maketitle

\begin{abstract}
A class of dynamic threshold models is proposed, for describing the upset of collective actions in social networks.
The agents of the network have to decide whether to undertake a certain action or not. They make their decision by comparing the activity level of their neighbors with a time-varying threshold, evolving according to a time-invariant opinion dynamic model. 
Key features of the model are a parameter representing the degree of self-confidence of the agents, and the mechanism adopted by the agents to evaluate the activity level of their neighbors. 
The case in which a radical agent, initially eager to undertake the action, interacts with a group of ordinary agents, is considered.
The main contribution of the paper is the complete analytic characterization of the asymptotic behaviors of the network, for three different graph topologies. The asymptotic activity patterns are determined as a function of the self-confidence parameter and of the initial threshold of the ordinary agents.
\end{abstract}

\section{Introduction} 
\label{sec:introduction}

A key problem in social sciences is that of understanding the complex relationships between the attitude of individuals and their collective behavior. 
The main challenge in this context is posed by modeling and analyzing the way in which the evolution of the individuals' opinions affects their will of undertaking or not a certain action.
In fact, it is widely recognized that this underlying mechanism is at the basis of crucial phenomena, such as the spread of behaviors and the arising of collective actions within social networks \cite{S09,C10,MS10}.  

Opinion dynamics is a well established research topic in the social science research field, which is receiving increasing attention from the control community (see, e.g., the recent survey \cite{Friedkin15}). 
Starting from the celebrated DeGroot model \cite{DeGroot74} and the numerous variations on it, a large body of literature has been developed in which the emphasis has been initially placed on the widely investigated consensus problem, which has an impact also on several other key problems in the control field (see \cite{Moreau05,RB05,OFM07} and references therein).
In recent years, researchers have concentrated their attention on models which present a richer variety of possible dynamic patterns, in order to describe the multiplicity of phenomena observed in social networks. Notable research lines along this path are the works on the Hegselmann-Krause model \cite{BHT09,YDH14,EB15}, the studies on the spread of misinformation \cite{AOP10,FZLJ12}, the analysis of networks with stubborn agents \cite{ACFO13,GS14,RFTI15} and the introduction of models including antagonistic interactions \cite{AL15}.

Using opinion dynamics models to predict behaviors of groups of individuals is a key problem in social psychology \cite{Friedkin10}. The problem is that to link the evolution of the individuals' opinion to their inclination about undertaking a certain action. In this respect, the simplest approach assumes the existence of a threshold, so that the individual becomes ``active'' whenever its opinion exceeds the threshold. This leads to the formulation of a so-called \emph{threshold model}.

Threshold models have been first introduced in \cite{G78}; since then, they have been employed to explain the collective behavior of a community of individuals in many different contexts: typical examples are the spread of technological innovations among large portions of the population; the attitude of masses towards new trends in popular culture; political phenomena such as riots, strikes, and so on. 
In this context, threshold models are well suited to predict the occurrence of cascade effects, i.e. the possibility that a behavior adopted by a small number of influential agents will propagate to a large part of the network \cite{ADO12}. In \cite{AOY11}, threshold models are adopted to analyze how innovations spread into a network starting from a set of promoters. Such a model has been later generalized in \cite{RG13} to account for the possibility for a member of the network to abandon a previously adopted innovation. Moreover the effect of the presence of a group of agents which maintain the innovation for a finite time despite the behavior of their neighbors is analyzed. 

In \cite{H14}, a threshold model has been adopted to describe the mechanisms underlying the formation of a collective action taking place during political unrest or social revolutions. In particular, the aim is to determine whether a radical agent is able to eventually persuade all the individuals of a network to engage in the demonstration: this occurs when the activity level of the individual's neighbors exceeds a certain threshold, which in turn evolves dynamically according to a classical DeGroot opinion model. The author uses this approach to analyze the effect of media interruption during the 2011 Egyptian revolution. Properties of this model have been studied in \cite{LHTM12}.

In this paper, starting from the model proposed in \cite{H14,LHTM12}, a more general class of threshold models is proposed and analyzed.
The main novelty is the introduction of a parameter which represents the relative confidence level that an agent has on her own opinion, with respect to that of her neighbors. This provides a new degree of freedom, which allows one to characterize the behavior of conservative networks, with respect to groups of individuals more inclined to change their attitude.
Another feature of the proposed model is the use of two different mechanisms for deciding whether an agent becomes active or not. Similarly to what is assumed in \cite{ADO12,RG13,LHTM12}, a non progressive model is adopted, meaning that each agent can change its actions multiple times, by comparing her current opinion with an indicator of the average activity level of her neighbors. In the proposed model, such an indicator can be either the fraction of active neighbors (as in \cite{LHTM12}), or a weighted average of the number of active neighbors which takes into account the self-confidence of each agent. The considered model can be also seen as an extension of the framework adopted in global games \cite{MS03,DTTZ13,TS14}, in which the agents make a decision upon undertaking an action according to a similar threshold-based mechanism, but once they have made their decision they do not change anymore their activity status.

The main contribution of the paper is the analysis of the asymptotic behavior of the network for different graph topologies. In particular, the complete analytic characterization of the asymptotic activity pattern of the network is determined for three topologies: the complete graph, the star graph and the ring graph. A wide variety of collective limiting behaviors is observed and its relationship with the self-confidence parameter and the chosen decision mechanism is highlighted. A preliminary version of this work has been presented in \cite{GGV15}. This paper extends previous results to the case of star and ring graphs.

The paper is organized as follows. In Section~\ref{sec:formulation}, the considered class of threshold models is introduced, together with the two different decision schemes. The analytic results on the network asymptotic behaviors are presented in Sections~\ref{sec:complete}, \ref{sec:star} and \ref{sec:ring}, for the cases of complete, star and ring graph topologies, respectively. \aggiungi{Results from numerical simulations carried out for a real ego network are reported in Section~\ref{sec:exp}.} Concluding remarks and future developments are provided in Section~\ref{sec:conclusions}.

\section{Problem formulation}
\label{sec:formulation}
A network of $n$ agents is described by an undirected graph $\mathcal{G} = (\mathcal{V},\mathcal{E})$, where $\mathcal{V}$ denotes the vertex set and $\mathcal{E} \subseteq \mathcal{V} \times \mathcal{V}$ is the edge set. Two agents $i$ and $j$ are \textit{neighbors} if $(i,j) \in \mathcal{E}$. Let $\mathcal{N}_i$ denote the set of neighbors of agent $i$ and $n_i$ be its cardinality. 
In this work, an agent is always considered a neighbor of itself, i.e. $(i,i)\in \mathcal{E}$ for all $i$, and the network topology is assumed to be time-invariant.

In order to model the agents' behavior, two variables are associated to agent $i$: the \emph{threshold} $\theta_i(t)\in [0,\ 1]$ and the \emph{action} $a_i(t)\in\{0,\ 1\}$. The variable $a_i$ discriminates whether the $i$th agent is undertaking a certain action at time $t$ ($a_i(t)=1$) or not ($a_i(t)=0$). The threshold $\theta_i(t)$ is used to model the attitude of the $i$th agent towards the possibility of becoming active. Depending on the context, it may represent the agent's opinion on a certain topic, or its intention to participate in some collective movement.
%\agtnote*{Mi sembra si possa togliere, viene detto prima di eq. (9)}{At time $t=0$, the action variable is $a_i(0)=1$ if $\theta_i(0)=0$ (agent $i$ is initially active) or $a_i(0)=0$ if $\theta_i(0) \neq 0$ (agent $i$ is initially active).} 

The agent behavior is described by the time evolution of the threshold and the action variables. At each time step, an agent updates its threshold to a weighted average of its neighbors' threshold
\begin{equation}
  \label{eq:thmodel}
  \theta_i(t+1) = \sum_{j \in \mathcal{N}_i} f_{ij} \theta_j(t), \qquad i=1,\dots,n,
\end{equation}
where the weights are such that $0 < f_{ij} <1$ and $\sum_{h}f_{ih} = 1$, $\forall i,j$.
Notice that \eqref{eq:thmodel} is the classic De Groot model, which has been widely employed in the literature on consensus and opinion dynamics \cite{DeGroot74}.

Besides updating the threshold, each agent computes the activity level of its neighbors as
\begin{equation}
  \label{eq:pmodel}
  p_i(t+1) = \sum_{j \in \mathcal{N}_i} g_{ij} a_j(t), \qquad i=1,\dots,n,
\end{equation}
%\agnote*{Qui mi tornerebbe comodo definirlo $p(t+1)$ invece di $p(t)$, ci sono controindicazioni?}{}
where $0 < g_{ij} <1$ and $\sum_{h}g_{ih} = 1$, $\forall i,j$. 
The action value of each agent is obtained by comparing the activity level $p_i(t)$ with the threshold $\theta_i(t)$, according to
\begin{equation}
  \label{eq:amodel}
  a_i(t) = \begin{cases}
    1 & \text{ if } p_i(t) \geq \theta_i(t) \\
    0 & \text{ else }
  \end{cases}, \quad i=1,\dots,n.
\end{equation}
By setting $f_{ij} = g_{ij} =0$ whenever $j \not\in \mathcal{N}_i$, equation \eqref{eq:thmodel} can be rewritten in matrix form as $\theta(t+1) = F \theta(t)$ and \eqref{eq:pmodel} becomes
\begin{equation}
%  \theta(t+1) &= F \theta(t),  \nonumber \\
  p(t+1)  = G a(t), \label{eq:modelm2} 
\end{equation}
where $\theta = [\theta_1,\dots,\theta_n]'$, $a = [a_i,\dots,a_n]'$, $p = [p_1,\dots,p_n]'$, and $F$ and $G$ are matrices whose $ij$-th entries are $f_{ij}$ and $g_{ij}$, respectively. By introducing the function 
$$
\phi(x) = \begin{cases}
  1 & \text{ if } x \geq 0, \\
  0 & \text{ else },
\end{cases}
$$
and exploiting~\eqref{eq:amodel},\eqref{eq:modelm2}, one gets
\begin{align}
  \theta(t+1) &= F \theta(t),  \label{eq:modelm1} \\
  a(t+1) &= \phi(G a(t) - F \theta(t)),   \label{eq:modelm3}
\end{align}
where the function $\phi(\cdot)$ is to be intended componentwise.

In this work, the entries of matrix $F$ are chosen as
\begin{equation}
  \label{eq:weights}
  f_{ij}  = \begin{cases}
    \frac {\beta} {\beta + n_i - 1} & \text{ if } i=j, \\
    \frac {1} {\beta + n_i - 1} & \text{ if } j\in\mathcal{N}_i,\ j \neq i,\\
    0 & \text{ else },
  \end{cases}
\end{equation}
where $\beta >0$ is the relative weight each agent assigns to its current threshold value compared to that of its neighbors. In other words, $\beta$ can be interpreted as the \emph{relative confidence} that each agent has on its own opinion, with respect to that of the other members of the network.
Two different ways of computing the neighbors' activity level are considered:
\begin{itemize}
\item[a)] \emph{Weighted Activity Level (WAL)}: in this setting $G=F$, i.e. the same relative weight is adopted both for computing the activity level of the neighbors and for weighting the neighbors' threshold;
\item[b)] \emph{Uniform Activity Level (UAL)}: in this scenario 
$$
g_{ij}  = \begin{cases}
    \frac {1} {n_i} & \text{ if } j\in\mathcal{N}_i,\\
    0 & \text{ else },
  \end{cases}
$$
so that $p_i(t+1)$ in \eqref{eq:pmodel} represents the fraction of neighbors of agent $i$ that are active at time $t$.
\end{itemize} 
Notice that in the UAL scenario each agent decides whether to become active or not by just ``counting'' the number of active neighbors. Conversely, in the WAL scenario  an agent weights in a different way the fact that its neighbors are active with respect to its own activity status. This is consistent with the idea that a self-confident individual, weighting its own opinion $\beta$ times that of its neighbors, will also consider in a different way its own behavior with respect to that of its neighbors. 
\begin{remark}
If $\beta=1$, one has $f_{ij}= g_{ij} = \frac 1 {n_i}$, $\forall i$, $\forall j\in \mathcal{N}_i$, and hence there is no difference between the two considered scenarios. The threshold update rule~\eqref{eq:modelm1} consists in computing the average of the neighbors' thresholds, and the activity level~\eqref{eq:modelm2} is equal to the fraction of active neighbors. Notice that in this special case, the setting considered in \cite{LHTM12} is recovered. 
\end{remark}

The objective of this this work is to study the asymptotic behavior of the nonlinear system \eqref{eq:modelm1},\eqref{eq:modelm3}, when the network initially contains one \emph{radical} agent (herafter labeled with index 1) and $n-1$ \emph{ordinary} agents. The radical agent is keen on undertaking an action and would like to convince the other agents to do the same: to this aim, at time $t=0$ its threshold is equal to zero and its activity variable is equal to 1. The ordinary agents are initially inactive and their threshold is equal to $\tau$, with $0 < \tau <1$. This corresponds to the initial condition
\begin{equation}
\label{eq:theta0}
\theta(0)=[0,\ \tau, \dots,\tau]',~~~ a(0)=[1,\ 0, \dots,0]'.
\end{equation}
Loosely speaking, $\tau$ represents the initial reluctance of the ordinary agents towards the action put forth by the radical one.
The problem addressed in the paper is to determine the asymptotic value of the action vector
$$
\ai=\lim_{t \rightarrow +\infty} a(t)
$$
under the initial condition \eqref{eq:theta0}, as a function of the initial reluctance of ordinary agents $\tau$ and of the relative confidence parameter $\beta$.
It is easy to check that $a_e=0$ and $a_e=\1$ (where $\1$ denotes a vector whose entries are all equal to 1) are always equilibria for system~\eqref{eq:modelm3}, irrespectively of the weighting matrices $F$ and $G$\footnote{With a slight abuse of notation, we assume that when $p=0$, then $a=0$, even if $\theta=0$. We do not modify the definition of $\phi$ in this sense, to keep notation simple.}.
However, several other equilibria may arise, which do depend on the topology of the interconnection network and on the values of $\beta$ and $\tau$. In the next sections, three different topologies will be analysed in detail.

\section{Complete graph}
\label{sec:complete}

Let the graph $\mathcal{G}$ be complete, i.e., $(i,j) \in \mathcal{E}$, for all $i,\ j$. The threshold evolution is characterized by the  following results.
\begin{lemma}
Consider the dynamic model~\eqref{eq:modelm1}, with $F$ chosen as in~\eqref{eq:weights}. If the interconnection graph is complete and $\theta(0) = [0,\ \tau, \dots,\tau]'$, $0<\tau < 1$, then
\begin{align}
\theta_1(t) &= \displaystyle{\frac{n-1}{n} \tau \left(1- \lambda_c^t \right),} \label{eq:thc1} \vspace*{1mm}\\
\theta_i(t) &= \displaystyle{\frac{n-1}{n} \tau \left(1+\frac{1}{n-1} \lambda_c^t \right),~i=2,\dots,n,} \label{eq:thc2}
\end{align}
where $\displaystyle{\lambda_c=\frac{\beta-1}{\beta+n-1}}$.
\label{lem:thc}
\vspace*{1mm}
\end{lemma}

\begin{proof}
When the graph is complete, from \eqref{eq:weights} one gets
$$
\begin{array}{rcl}
F &=& \displaystyle{ \frac {\beta} {\beta + n- 1} I_n +  \frac {1} {\beta + n - 1}  (\1\1'-I_n) } \vspace*{1mm}\\
  &=& \displaystyle{\frac {\beta-1} {\beta + n- 1} I_n +  \frac {n} {\beta + n - 1}  \frac{\1\1'}{n} }\vspace*{1mm}\\
  &=&  \displaystyle{\frac{\1\1'}{n} + \frac {\beta-1} {\beta + n- 1} \left(I_n -  \frac{\1\1'}{n} \right). }
\end{array}
$$
Then, it is easy to check that, for every $t \geq 0$, one has
$$
F^t =    \frac{\1\1'}{n} + \left(\frac {\beta-1} {\beta + n- 1}\right)^t \left(I_n -  \frac{\1\1'}{n} \right).
$$
For the initial condition $\theta(0) = [0,\ \tau, \dots,\tau]'$ one gets
$$
\begin{array}{rcl}
\theta(t) &=& F^t \theta(0) \\
&=& \displaystyle{ \frac{n-1}{n} \tau\, \1 + \left( \frac {\beta-1} {\beta + n- 1}\right)^t \left( \theta(0)-\frac{n-1}{n} \tau \1\right) }
\end{array}
$$
from which \eqref{eq:thc1}-\eqref{eq:thc2} immediately follow.
\end{proof}

\begin{corollary}
\label{cor:thc}
  Consider the dynamic model~\eqref{eq:modelm1}, with the weights chosen as in~\eqref{eq:weights}. If the interconnection graph is complete and $\theta(0) = [0,\ \tau, \dots,\tau]'$, $0<\tau < 1$, then
\begin{equation}
\label{eq:cor0}
\lim_{t \to \infty} \theta(t) = \frac {n-1} n \tau \1.
\end{equation}
Moreover, if $\beta>1$, then 
\begin{align}
  \theta_1(t+1) &> \theta_1(t), \label{eq:cor1} \\
  \theta_i(t+1) &< \theta_i(t), \quad i=2,\dots,n, \label{eq:cor2}
\end{align}
for all $t\geq0$.
\end{corollary}

\subsection{Weighted activity level}

Let us consider the WAL setting first, i.e.,
\begin{equation}
\label{eq:GeqFc}
G=F= \frac {\beta} {\beta + n- 1} I_n +  \frac {1} {\beta + n - 1}  (\1\1'-I_n).
\end{equation}
Define the functions of $\beta$:
\begin{eqnarray}
\gamma_1(\beta) &=&  \frac{n}{n-1} \, \frac{1}{\beta+n-1}, \label{eq:gamma1}\\
\gamma_2(\beta) &=&  \frac{n}{n-1} \, \frac{\beta}{\beta+n-1}, \label{eq:gamma2}\\
\gamma_3(\beta) &=& \frac{1}{\beta+n-2}.\label{eq:gamma3}
\end{eqnarray}
Such functions are shown in Figures \ref{fig:n5GeqF}-\ref{fig:n20GeqF} for $n=5$ and $n=20$, respectively. The following result holds.
\begin{theorem}
\label{th:FGa}
System  \eqref{eq:modelm1},\eqref{eq:modelm3}, with $G=F$ given by \eqref{eq:GeqFc} and initial condition \eqref{eq:theta0}, exhibits the following asymptotic behaviors.
\begin{itemize}
\item[i)] For $\beta > 1$,
\begin{itemize}
\item[-] if $\tau < \gamma_1(\beta)$, then $\ai = \1$; \vspace*{2mm}
\item[-] if $\tau > \gamma_2(\beta)$, then $\ai = 0$; \vspace*{2mm}
\item[-] if $\gamma_1(\beta) \leq \tau \leq \gamma_2(\beta)$, then $a(t)=a(0)$, $\forall t \geq 1$.
\end{itemize}
\item[ii)] For $\beta \leq 1$, if $\tau \leq \gamma_3(\beta)$, then $\ai = \1$; otherwise $\ai = 0$.
\end{itemize}
\end{theorem}

\begin{proof}
According to Lemma \ref{lem:thc}, one has $\theta_2(t)=\theta_3(t)=\dots=\theta_n(t)$, $\forall t\geq 0$.
Since from \eqref{eq:theta0} it also holds $a_2(0)=a_3(0)=\dots=a_n(0)$, this clearly implies  $a_2(t)=a_3(t)=\dots=a_n(t)$, $\forall t\geq 0$. Hence, in the sequel we will refer only to $\theta_2(t)$ and $a_2(t)$.
\\
i) Given the initial condition \eqref{eq:theta0}, from Lemma~\ref{lem:thc} one has
\begin{align}
    \theta_1(1) &= \frac{n-1}{\beta+n-1} \,\tau, \vspace*{1mm} \label{eq:th11a}\\
    \theta_2(1) &=  \frac{\beta + n-2}{\beta+n-1} \, \tau. \label{eq:th21a}
\end{align}
Being $p(1)=Ga(0)$, one gets
\begin{align*}
    p_1(1) &= \frac{\beta}{\beta+n-1},  \vspace*{1mm}\\
    p_2(1) &=  \frac{1}{\beta+n-1},
\end{align*}
and hence $a(1)=\1$ if and only if the following conditions are satisfied
\begin{align}
    \beta & \geq (n-1) \tau,  \vspace*{1mm} \label{eq:1in1a1}\\
    1 & \geq  (\beta + n -2 )\,\tau. \label{eq:1in1a2}
\end{align}
Similarly, when both conditions \eqref{eq:1in1a1}-\eqref{eq:1in1a2} are violated, one has $a(1)=0$. 
When $\frac{1}{\beta+n-2} < \tau \leq \frac{\beta}{n-1}$, one has $a(1)=a(0)$. 
If $\beta>1$, due to Corollary~\ref{cor:thc} one has that $\theta_1(t)$ is increasing, while  $\theta_2(t)$ is decreasing. 
Therefore, one will have $a(t)=a(0)$ and $p(t)=p(0)$ until either $p_1(t)<\theta_1(t)$ or $p_2(t)\geq \theta_2(t)$. From~\eqref{eq:cor0}, both conditions will never occur if the following inequalities hold
\begin{align}
  \frac{\beta}{\beta+n-1}&\geq \frac{n-1}{n} \, \tau, \label{eq:condiva} \vspace*{1mm} \\
  \frac{1}{\beta+n-1} &\leq \frac{n-1}{n} \, \tau, \label{eq:condivb}
\end{align}
which correspond to $\gamma_1(\beta) \leq \tau \leq \gamma_2(\beta)$.
If $\tau > \gamma_2(\beta)$, \eqref{eq:condiva}  does not hold: therefore, $p_1(t)<\theta_1(t)$ for some $t$ and hence $\ai=0$.
Similarly, if $\tau < \gamma_1(\beta)$, \eqref{eq:condivb}  does not hold, thus leading to $\ai=\1$.
\\
ii) From \eqref{eq:1in1a1}-\eqref{eq:1in1a2} we have that $a(1)=\1$ if and only if $\tau \leq \min \left\{ \frac{\beta}{n-1}, \frac{1}{\beta+n-2} \right\}$, while
$a(1)=0$ if and only if $\tau > \max \left\{ \frac{\beta}{n-1}, \frac{1}{\beta+n-2} \right\}$.
Since $\beta \leq 1$ implies $\frac{\beta}{n-1} \leq \frac{1}{\beta+n-2}$ it remains to discuss the case in which
\begin{equation}
\label{eq:assiii}
\frac{\beta}{n-1}  < \tau \leq  \frac{1}{\beta+n-2},
\end{equation} 
which, according to the above discussion, leads to $a(1)=[0, \ 1,\dots,1]'$.
Hence $p(2)=Ga(1)$ is such that
\begin{align*}
    p_1(2) &= \frac{n-1}{\beta+n-1} ,  \vspace*{1mm}\\
    p_2(2) &=  \frac{\beta+n-2}{\beta+n-1} .
\end{align*}
Being from Lemma~\ref{lem:thc}
%\begin{align}
%    \theta_1(2) &= \frac{n-1}{n} \,\tau \, \left\{ 1-\left(\frac{\beta-1}{\beta+n-1}\right)^2\right\}, \vspace*{1mm} \\
%    \theta_2(2) &=  \frac{n-1}{n} \,\tau \, \left\{ 1+\frac{1}{n-1}\left(\frac{\beta-1}{\beta+n-1}\right)^2\right\} ,
%\end{align}
\begin{align*}
    \theta_1(2) &= \frac{n-1}{n} \,\tau \, \left( 1-\lambda_c^2\right), \vspace*{1mm} \\
    \theta_2(2) &=  \frac{n-1}{n} \,\tau \, \left( 1+\frac{1}{n-1}\lambda_c^2\right) ,
\end{align*}
through long but straightforward calculations it is possible to verify that, under the assumptions \eqref{eq:assiii}, one has
$$
p_1(2)\geq \theta_1(2),~~ p_2(2)\geq \theta_2(2).
$$
Therefore, $a(2)=\1$ and one can conclude that $\ai=\1$ for every $\tau \leq \frac{1}{\beta+n-2}=\gamma_3(\beta)$. 
%This concludes the proof.
\end{proof}

\begin{figure}[t]
\centering
\psfrag{n}{$n=5$}
\psfrag{beta}[t]{$\beta$}
\psfrag{tau}{$\tau$}
\psfrag{titulo}{}
\psfrag{g1}{$\gamma_1(\beta)$}
\psfrag{g2}{$\gamma_2(\beta)$}
\psfrag{g3}{$\gamma_3(\beta)$}
\includegraphics[width=0.65\columnwidth]{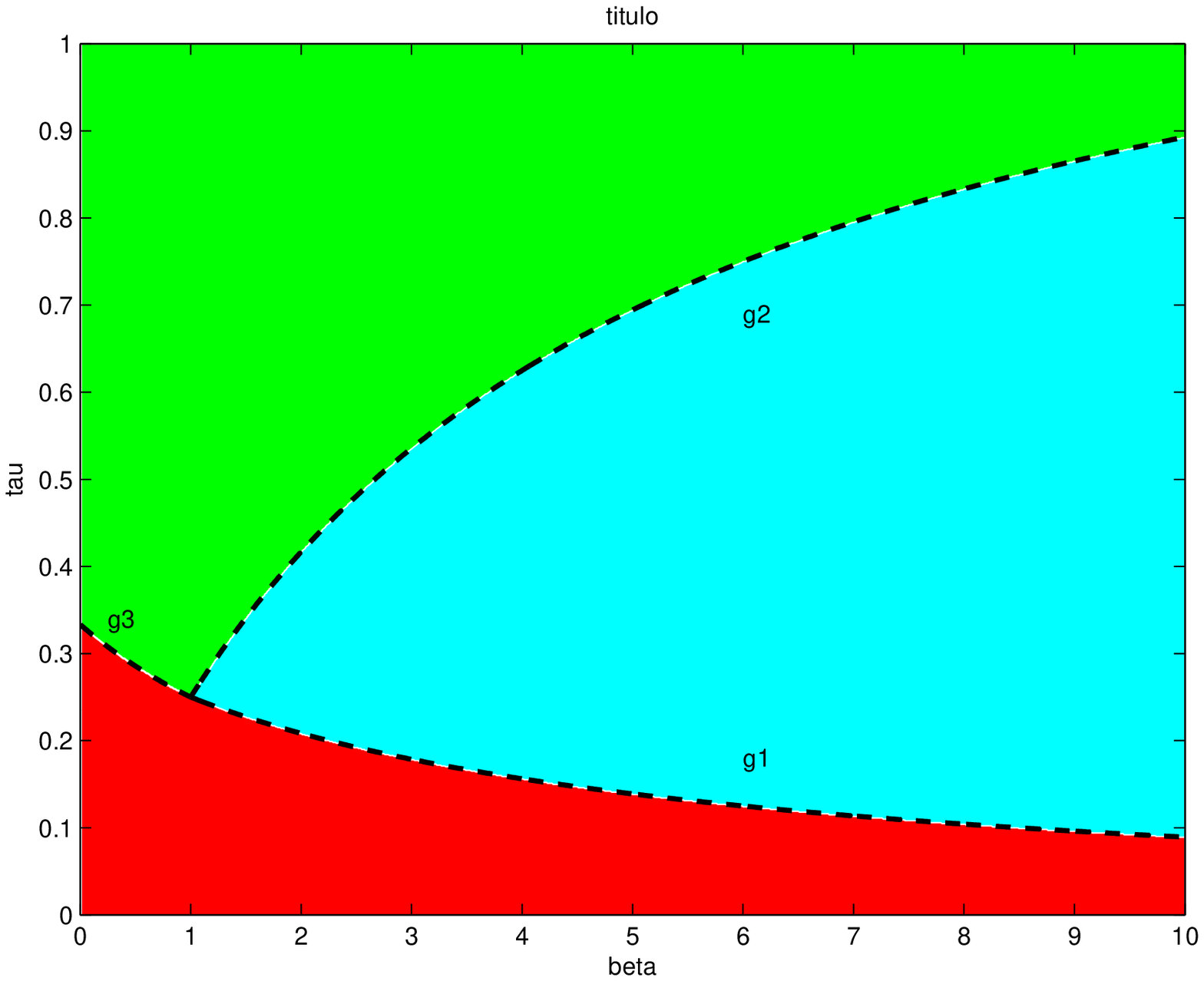}
\caption{Complete graph, WAL setting, $n=5$.}
\label{fig:n5GeqF}
\end{figure}
\begin{figure}[htb]
\centering
\psfrag{n}{$n=20$}
\psfrag{beta}[t]{$\beta$}
\psfrag{tau}{$\tau$}
\psfrag{titulo}{}
\psfrag{g1}{$\gamma_1(\beta)$}
\psfrag{g2}{$\gamma_2(\beta)$}
\psfrag{g3}{$\gamma_3(\beta)$}
\includegraphics[width=0.65\columnwidth]{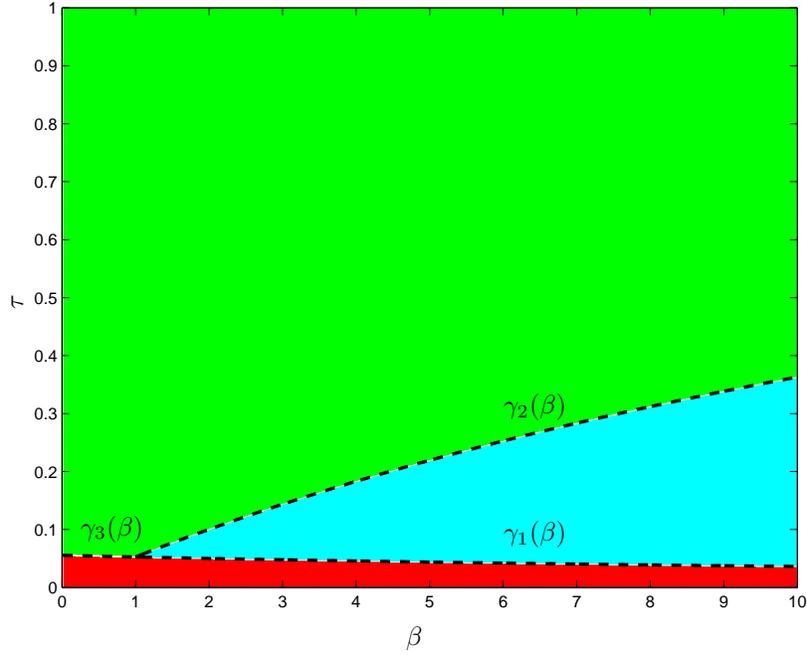}
\caption{Complete graph, WAL setting, $n=20$.}
\label{fig:n20GeqF}
\end{figure}

A byproduct of the proof of Theorem \ref{th:FGa} is the characterization of the cases for which the asymptotic behavior is reached in one step.
\begin{corollary}
\label{cor:compa}
System  \eqref{eq:modelm1},\eqref{eq:modelm3}, with $G=F$ given by  \eqref{eq:GeqFc} and initial condition \eqref{eq:theta0}.
satisfies
\begin{itemize}
\item[-] $a(t)=\1$ for all $t \geq 1$, if and only if  $\tau \leq \min \left\{ \frac{\beta}{n-1},\gamma_3(\beta) \right\}$;
\item[-] $a(t)=0$ for all $t \geq 1$, if  and only if $\tau > \max \left\{ \frac{\beta}{n-1}, \gamma_3(\beta) \right\}$.
\end{itemize}
\end{corollary}

Theorem \ref{th:FGa} gives the complete characterization of the asymptotic behavior of system  \eqref{eq:modelm1},\eqref{eq:modelm3}, with initial condition \eqref{eq:theta0}. Notice that there are three possible asymptotic activity profiles: i) all the agents become active; ii) all the agents become inactive; iii) the situation remains always the same as in the initial condition (i.e., agent 1 is active and all the others are inactive). 
From the proof of Theorem \ref{th:FGa} it is apparent that the asymptotic value $\ai$ is always reached in a finite number of steps. However, it is not possibile to give an a priori upper bound to such a number, which can be arbitrarily high. For example,  if $n=5$, $\tau=0.99$ and $\beta=15$, one has that $a(t)=0$ only for $t \geq 19$.  Similarly, if $n=5$, $\tau=0.01$ and $\beta=118$, one has that $a(t)=\1$ only for $t \geq 56$. 

Figures \ref{fig:n5GeqF} and \ref{fig:n20GeqF} show the asymptotic behaviors achieved for different values of  $\beta$ (relative confidence parameter) and $\tau$ (initial threshold of the ordinary agents), in the cases of $5$ and $20$ agents, respectively. 
The different regions correspond to: $\ai=\1$ (red); $\ai=0$ (green); $a(t)=a(0)$, $\forall t$ (light blue). 
The dashed curves represent the functions $\gamma_i(\beta)$ defined in \eqref{eq:gamma1}-\eqref{eq:gamma3}. Notice that these curves intersect at $\beta=1$. It can be observed that for $n=20$ the curves $\gamma_1(\tau)$ and $\gamma_3(\tau)$ are almost indistinguishable. As expected, the area in which all the agents end up to be inactive grows with $n$, while the region in which all the agents become active tends to shrink, as well as that in which the initial condition $a(0)$ is maintained indefinitely.

\subsection{Uniform activity level}

Now, let us consider setting UAL. When the interconnection graph is complete, this means that
\begin{equation}
\label{eq:Gbc}
G=\frac {\1 \1'} {n}, 
\end{equation}
and $F$ is given by \eqref{eq:GeqFc}.
Let us define the function
$$
\eta(\beta) = \frac{\beta+n-1}{n (\beta+n-2)}.
$$
Then, the following result holds.
\begin{theorem}
\label{th:FGb}
System  \eqref{eq:modelm1},\eqref{eq:modelm3}, with $F$ defined as in \eqref{eq:GeqFc}, $G$ given by \eqref{eq:Gbc} and initial condition \eqref{eq:theta0}, has the following asymptotic behaviors:
\begin{itemize}
\item[i)] For $\beta > 1$,
\begin{itemize}
\item[-] if $\tau < \frac{1}{n-1}$, then $\ai = \1$; \vspace*{2mm}
\item[-] if $\tau > \frac{1}{n-1}$, then $\ai = 0$; \vspace*{2mm}
\item[-] if $\tau = \frac{1}{n-1}$, then $a(t)=a(0)$, $\forall t \geq 1$.
\end{itemize}
\item[ii)] For $\beta \leq 1$, if $\tau \leq \eta(\beta)$, then $\ai = \1$; otherwise $\ai = 0$.
\end{itemize}
\end{theorem}

\begin{proof}
i) By following the same reasoning as in the proof of Theorem \ref{th:FGa}, one gets \eqref{eq:th11a}-\eqref{eq:th21a} and, being $G$ given by \eqref{eq:Gbc},
$p_1(1) =p_2(1)= \frac{1}{n}$.
Hence, $a(1)=\1$ if and only if the following conditions are satisfied
\begin{align*}
    \frac{1}{n} & \geq \frac{n-1}{\beta+n-1} \, \tau,  \vspace*{1mm} \\
        \frac{1}{n} & \geq  \frac{\beta + n -2 }{\beta+n-1} \,\tau, 
\end{align*}
while $a(1)=0$ if and only if both conditions are violated.
When $\frac{\beta+n-1}{n(\beta+n-2)} < \tau \leq \frac{\beta+n-1}{n(n-1)}$, one has $a(1)=a(0)$. 
Notice that this can occur only if $\beta>1$, which means that \eqref{eq:cor1}-\eqref{eq:cor2} in Corollary~\ref{cor:thc} hold.
Hence, $a(t)=a(0)$ until either $p_1(t)<\theta_1(t)$ or $p_2(t)\geq \theta_2(t)$, are verified. Since $p(t+1)=Ga(t)=\frac{1}{n}\1$, these conditions correspond respectively to
\begin{align}
  \frac{1}{n}&< \frac{n-1}{n} \, \tau \, \left(1-\lambda_c^t \right), \label{eq:condiva2} \vspace*{1mm} \\
   \frac{1}{n} &\geq \frac{n-1}{n} \, \tau \left(1+\frac{1}{n-1}\lambda_c^t \right), \label{eq:condivb2}
\end{align}
where $\lambda_c=\frac{\beta-1}{\beta+n-1}$ as in Lemma~\ref{lem:thc}. Being $0 < \lambda_c <1$, \eqref{eq:condiva2}-\eqref{eq:condivb2} lead respectively to
%\begin{align}
%  t &> \log_{\lambda_c} \left( 1-  \frac{1}{(n-1)\tau} \right) \label{eq:timea} \vspace*{1mm} \\
%   t &\geq \log_{\lambda_c} \left(  \frac{1}{\tau} -n+1 \right).\label{eq:timeb}
%\end{align}
\begin{align*}
  \lambda_c^t &<  1-  \frac{1}{(n-1)\tau},   \vspace*{1mm} \\
  \lambda_c^t &\leq    \frac{1}{\tau} -n+1~ .
\end{align*}
Hence, one eventually gets $a(t) = 0$, for some $t$, whenever 
$$
1-  \frac{1}{(n-1)\tau} > \frac{1}{\tau} -n+1,
$$
which corresponds to $\tau > \frac{1}{n-1}$.
Conversely, if $\tau < \frac{1}{n-1}$, \eqref{eq:condivb2} will occur before \eqref{eq:condiva2}, thus leading to $a(t) = \1$.
Finally, for $\tau = \frac{1}{n-1}$, \eqref{eq:condiva2}-\eqref{eq:condivb2} are never satisfied 
%(the times in \eqref{eq:timea}-\eqref{eq:timeb} diverge) 
and therefore $a(t)=a(0)$ indefinitely.
\\
ii) Let $\beta\leq1$. Similarly to the proof of item ii) in Theorem \ref{th:FGa}, it is possible to show that if 
$$
\frac{\beta+n-1}{n(n-1)} < \tau \leq \frac{\beta+n-1}{n(\beta+n-2)}
$$
one gets $a(1)=[0, \ 1,\dots,1]'$ and, after long but straightforward manipulations, 
$$
p_1(2)\geq \theta_1(2),~~ p_2(2)\geq \theta_2(2).
$$
Therefore, $a(2)=\1$ and hence $\ai=1$ for every $\tau \leq \frac{\beta+n-1}{n(\beta+n-2)}=\eta(\beta)$. 
\end{proof}

\begin{figure}[t]
\centering
\psfrag{n}{$n=5$}
\psfrag{titulo}{}
\psfrag{beta}[t]{$\beta$}
\psfrag{tau}{$\tau$}
\psfrag{eta}{$\eta(\beta)$}
\psfrag{t}{$\tau=\frac{1}{n-1}$}
\includegraphics[width=0.65\columnwidth]{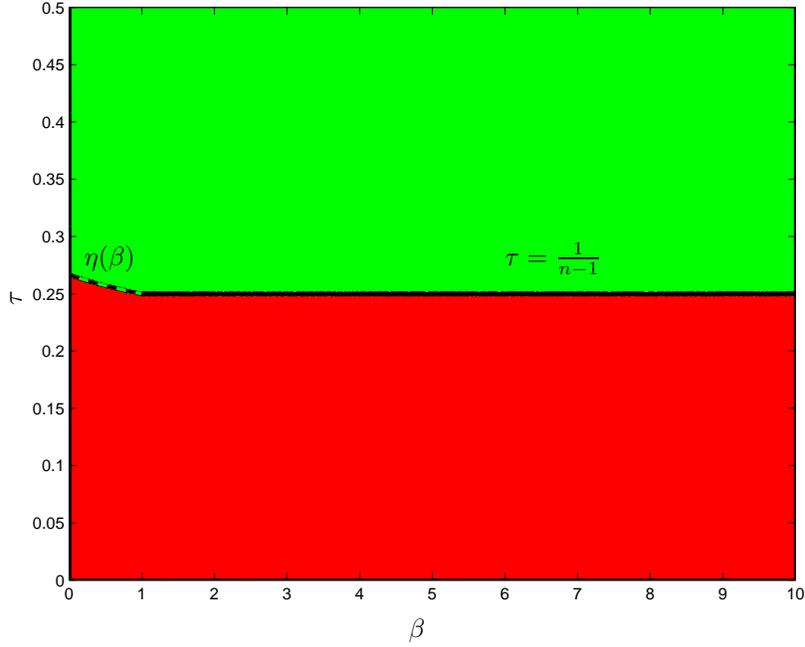}
\caption{Complete graph, UAL setting, $n=5$.}
\label{fig:n5GneF}
\end{figure}
\begin{figure}[t]
\centering
\psfrag{n}{$n=20$}
\psfrag{titulo}{}
\psfrag{beta}[t]{$\beta$}
\psfrag{tau}{$\tau$}
\psfrag{eta}{$\eta(\beta)$}
\psfrag{t}{$\tau=\frac{1}{n-1}$}
\includegraphics[width=0.65\columnwidth]{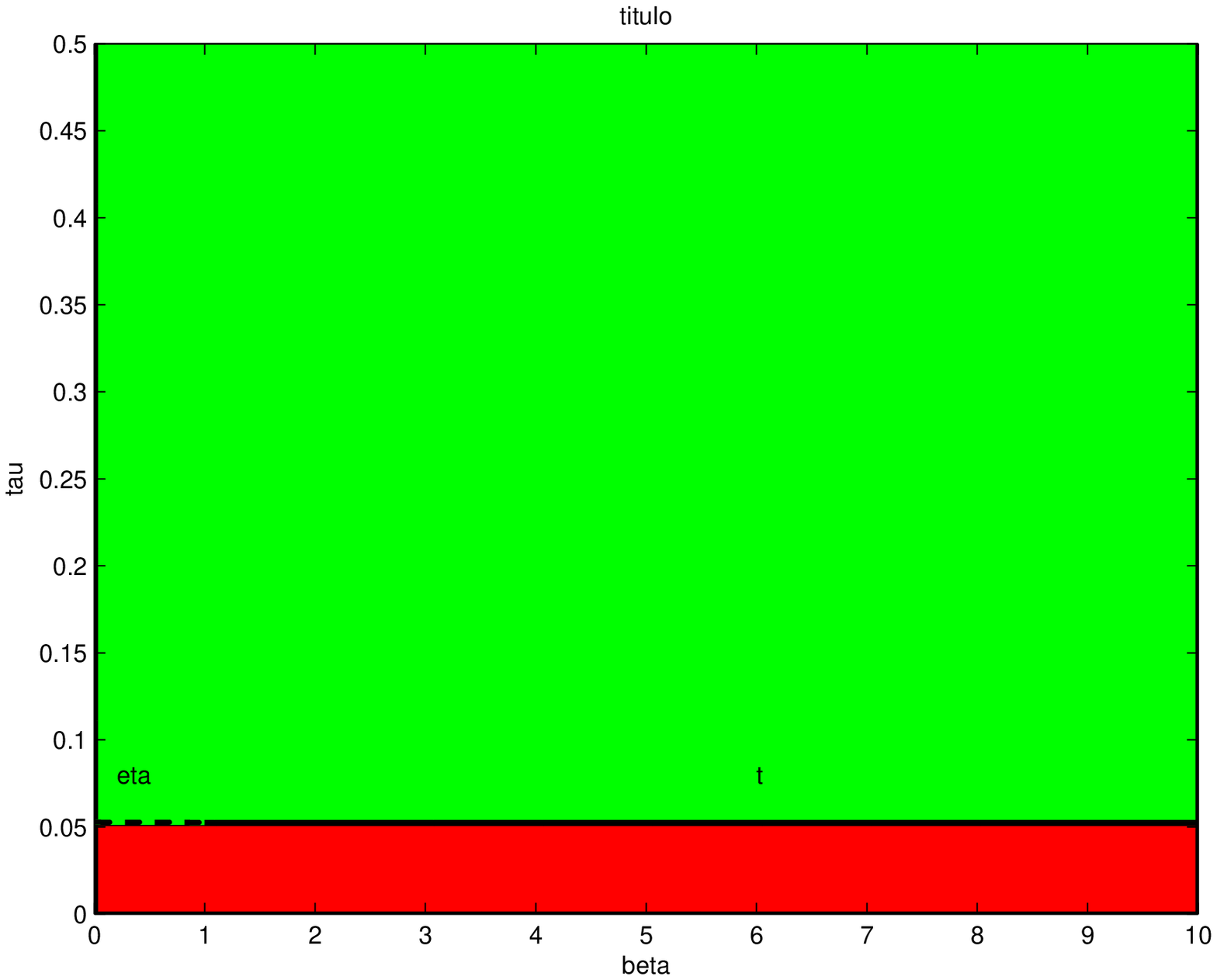}
\caption{Complete graph, UAL setting, $n=20$.}
\label{fig:n20GneF}
\end{figure}

\begin{corollary}
\label{cor:compb}
System  \eqref{eq:modelm1},\eqref{eq:modelm3}, with $F$ defined as in \eqref{eq:GeqFc}, $G$ given by \eqref{eq:Gbc} and initial condition \eqref{eq:theta0}, satisfies
\begin{itemize}
\item[-] $a(t)=\1$ for all $t \geq 1$, if and only if  $\tau \leq \min \left\{ \frac{\beta+n-1}{n(n-1)},\eta(\beta) \right\}$;
\item[-] $a(t)=0$ for all $t \geq 1$, if  and only if $\tau > \max \left\{ \frac{\beta+n-1}{n(n-1)}, \eta(\beta) \right\}$.
\end{itemize}
\end{corollary}

Figures \ref{fig:n5GneF} and \ref{fig:n20GneF} show the curve $\eta(\beta)$ in the $\tau-\beta$ plane (dashed), along with the horizontal line $\tau=\frac{1}{n-1}$ (solid), for $n=5$ and $n=20$ respectively. Notice that in the latter case, the two lines are almost coincident. Also in this setting, the area in which all the agents end up to be inactive grows with $n$ (notice the scale on $\tau$), while the region in which all the agents become active tends to shrink. 

It is worth observing that when $\beta=1$, matrix $G$ is the same in both the WAL and the UAL setting, so that conditions in Theorems~\ref{th:FGa} and \ref{th:FGb} coincide. In this case, the scenario addressed in \cite{LHTM12} is recovered. In particular,  from Corollaries \ref{cor:compa} and \ref{cor:compb} it turns out that only two situations occur: either $a(1)=\1$ if $\tau \leq \frac{1}{n-1}$, or $a(1)=0$ otherwise. Hence, the steady state behavior is always achieved in one step. 
The introduction of the parameter $\beta$, accounting for the relative confidence of each agent on its own opinion, has significantly enriched the picture of possible asymptotic behaviors of the system. For $\beta >1$, three new different situations appear in the WAL setting: all the agents eventually become active; all the agents eventually become inactive; the initial situation is maintained indefinitely. As $\beta$ increases, the latter situation occurs for a larger range of values of the initial threshold $\tau$. This corresponds to the fact that in a network whose agents are more self-confident, it is more difficult to persuade them to change their status. Conversely, for $\beta<1$, this behavior disappears and either $\1$ is reached (in one or two steps) or all the agents become inactive in one step.

Another interesting observation concerns the differences between the WAL and UAL scenarios. The same five behaviors described above for the WAL setting, are present also in setting UAL, but the condition in which $a(t)=a(0)$, $\forall t$, occurs only if $\tau$ is exactly equal to $\frac{1}{n-1}$, which is clearly a singular condition. 

\section{Star graph}
\label{sec:star}

In this section we analyze the asymptotic behavior of system \eqref{eq:modelm1},\eqref{eq:modelm3} when the graph has a star structure, as depicted in Figure \ref{fig:star2}.
\begin{figure}[tb]
\centering
\psfrag{1}{$1$}
\psfrag{2}{$2$}
\psfrag{3}{$3$}
\psfrag{4}{$4$}
\psfrag{5}{$5$}
\psfrag{...}{$\dots$}
\psfrag{n}{$n$}
\includegraphics[width=0.4\columnwidth]{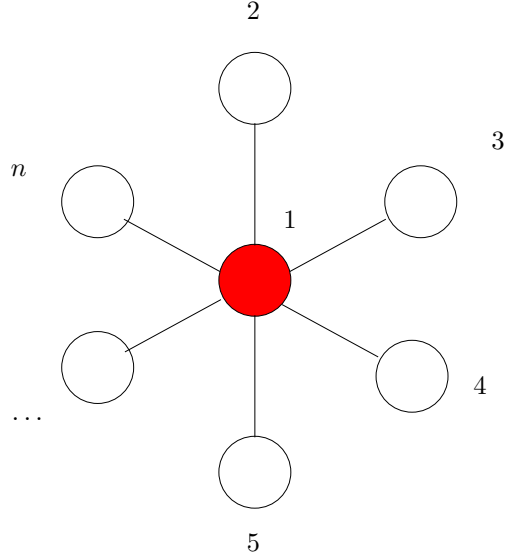}
\caption{Star graph with radical agent in the center.}
\label{fig:star2}
\end{figure}
In the star graph, the radical agent ($i=1$) is the center of the graph, while the remaining $n-1$ ordinary agents are connected only to the radical. This leads to a matrix $F$ of the form
\begin{equation}
\label{eq:Fstar}
F=
%\left(
%\begin{array}{ccccc}
\begin{bmatrix}
\frac{\beta}{\beta+n-1} & \frac{1}{\beta+n-1} & \dots & \dots & \frac{1}{\beta+n-1} \\
\frac{1}{\beta+1} & \frac{\beta}{\beta+1}  & 0 & \dots & 0 \\
\vdots & 0 & \ddots &  & 0 \\
\vdots & \vdots &  & \ddots & 0 \\
\frac{1}{\beta+1} & 0 & \dots & \dots & \frac{\beta}{\beta+1}  
\end{bmatrix}
%\end{array}
%\right).
\end{equation}
We consider scenarios WAL and UAL, defined as in Section \ref{sec:complete}. While in the former $G=F$, in the latter one has
\begin{equation}
\label{eq:Gstar}
G=
%\left(
%\begin{array}{ccccc}
\begin{bmatrix}
\frac{1}{n} & \frac{1}{n} & \dots & \dots & \frac{1}{n} \\
\frac{1}{2} & \frac{1}{2}  & 0 & \dots & 0 \\
\vdots & 0 & \ddots &  & 0 \\
\vdots & \vdots &  & \ddots & 0 \\
\frac{1}{2} & 0 & \dots & \dots & \frac{1}{2}  
\end{bmatrix}
%\end{array}
%\right).
\end{equation}
Let us define
\begin{align}
r &= \frac{(n-1)(\beta+1)}{\beta+n-1}, \label{eq:ths}\\
\lambda_s &= \frac{\beta(2\beta+n)}{(\beta+1)(\beta+n-1)}-1. \label{eq:lams}
\end{align}
The following technical results are instrumental to the asymptotic analysis of the star graph interconnection.
\begin{lemma}
Consider the dynamic model~\eqref{eq:modelm1}, with matrix $F$ given by \eqref{eq:Fstar}. If $\theta(0) = [0,\ \tau, \dots,\tau]'$, $0<\tau < 1$, then
\begin{align}
\theta_1(t) &= \displaystyle{\frac{r}{1+r} \tau \left(1- \lambda_s^t \right),} \label{eq:ths1} \vspace*{1mm}\\
\theta_i(t) &= \displaystyle{\frac{r}{1+r} \tau \left(1+ \frac{1}{r} \,\lambda_s^t \right),~i=2,\dots,n,} \label{eq:ths2}
\end{align}
where $r$ and $\lambda_s$ are given by \eqref{eq:ths} and \eqref{eq:lams}, respectively. 
\label{lem:star} 
\vspace*{1mm}
\end{lemma}

\begin{proof}
Let us first observe that, due to the structure of $F$ in \eqref{eq:Fstar}, one has that system \eqref{eq:modelm1} with the initial condition $\theta(0) = [0,\ \tau, \dots,\tau]'$ satisfies $\theta_2(t)=\theta_3(t)=\dots=\theta_n(t)$, $\forall t\geq 0$. This implies that one can analyze the behavior of $\theta(t)$ by considering the two-dimensional system
$$
\begin{array}{rcl}
\theta_1(t+1) &=& \displaystyle{ \frac {\beta} {\beta + n- 1} \theta_1(t) +  \frac {n-1} {\beta + n - 1} \theta_2(t) } \vspace*{1mm}\\
\theta_2(t+1)  &=& \displaystyle{\frac {1} {\beta + 1} \theta_1(t) +  \frac {\beta} {\beta + 1}  \theta_2(t).}
\end{array}
$$
It is easy to check that the eigenvalues of such system are $1$ and $\lambda_s$ in \eqref{eq:lams}. Moreover, \eqref{eq:ths1}-\eqref{eq:ths2} readily follow from the system response to the initial condition $\theta(0)$.
\end{proof}

\begin{corollary}
\label{cor:star}
Consider the dynamic model~\eqref{eq:modelm1}, with matrix $F$ given by \eqref{eq:Fstar}. If $\theta(0) = [0,\ \tau, \dots,\tau]'$, $0<\tau < 1$, then
\begin{equation}
\label{eq:cor0star}
\lim_{t \to \infty} \theta(t) = \frac{r}{1+r}  \tau \,  \1.
\end{equation}
Moreover, if $\beta>\sqrt{n-1}$, then $\lambda_s >0$ and
\begin{align*}
  \theta_1(t+1) &> \theta_1(t),  \\
  \theta_i(t+1) &< \theta_i(t), \quad i=2,\dots,n, 
\end{align*}
for all $t\geq0$.
\end{corollary}

\subsection{Weighted activity level}

Let us first, consider the WAL setting, i.e., $G=F$ as in \eqref{eq:Fstar}.
Define the functions of $\beta$:
\begin{align}
\delta_1(\beta) &=  \frac{n \beta + 2(n-1)}{(n-1)(\beta+1)^2}, \label{eq:delta1}\\
\delta_2(\beta) &=  \frac{\beta(\beta+1)}{(\beta+n-1)} \, \delta_1(\beta), \label{eq:delta2}\\
\delta_3(\beta) &= \beta \, \delta_1(\beta).\label{eq:delta3}
\end{align}
Such functions are shown in Figures \ref{fig:star_a5}-\ref{fig:star_a20} for $n=5$ and $n=20$, respectively. The following result holds.
\begin{theorem}
\label{th:star_a}
System  \eqref{eq:modelm1},\eqref{eq:modelm3}, with $G=F$ given by \eqref{eq:Fstar} and initial condition \eqref{eq:theta0}, exhibits the following asymptotic behaviors.
\begin{itemize}
\item[i)] For $\beta > \sqrt{n-1}$,
\begin{itemize}
\item[-] if $\tau < \delta_1(\beta)$, then $\ai = \1$; \vspace*{2mm}
\item[-] if $\tau > \delta_2(\beta)$, then $\ai = 0$; \vspace*{2mm}
\item[-] if $\delta_1(\beta) \leq \tau \leq \delta_2(\beta)$, then $a(t)=a(0)$, $\forall t \geq 1$.
\end{itemize}
\item[ii)] For  $1 \leq \beta \leq \sqrt{n-1}$, if
\begin{equation}
\label{eq:fractal}
\left\lceil
\log_{\lambda_s^2} \frac{1}{\lambda_s} \left( \frac{1}{\tau} \frac{r+1}{\beta+1} - r \right)
\right\rceil
\geq
\left\lceil
\log_{\lambda_s^2} \left( \frac{\beta}{\tau} \frac{r+1}{\beta+1} - r \right)
\right\rceil
\end{equation}
then $\ai=\1$. Otherwise, $\ai = 0$.
\item[iii)] For $\beta <1$,
\begin{itemize}
\item[-] if $\tau \leq \delta_3(\beta)$, then $\ai = \1$; \vspace*{2mm}
\item[-] if $\tau > \delta_1(\beta)$, then $\ai = 0$; \vspace*{2mm}
\item[-] if $\delta_3(\beta)< \tau \leq \delta_1(\beta)$, then $a(t)$ oscillates indefinitely between $[1,\ 0, \dots,0]'$ and $[0,\ 1, \dots,1]'$.
\end{itemize}
\end{itemize}
\end{theorem}

\begin{figure}[tb]
\centering
\psfrag{beta}[t]{$\beta$}
\psfrag{tau}{$\tau$}
\psfrag{d1}{$\delta_1(\beta)$}
\psfrag{d2}{$\delta_2(\beta)$}
\psfrag{d3}{$\delta_3(\beta)$}
\psfrag{titulo}{}
\includegraphics[width=0.65\columnwidth]{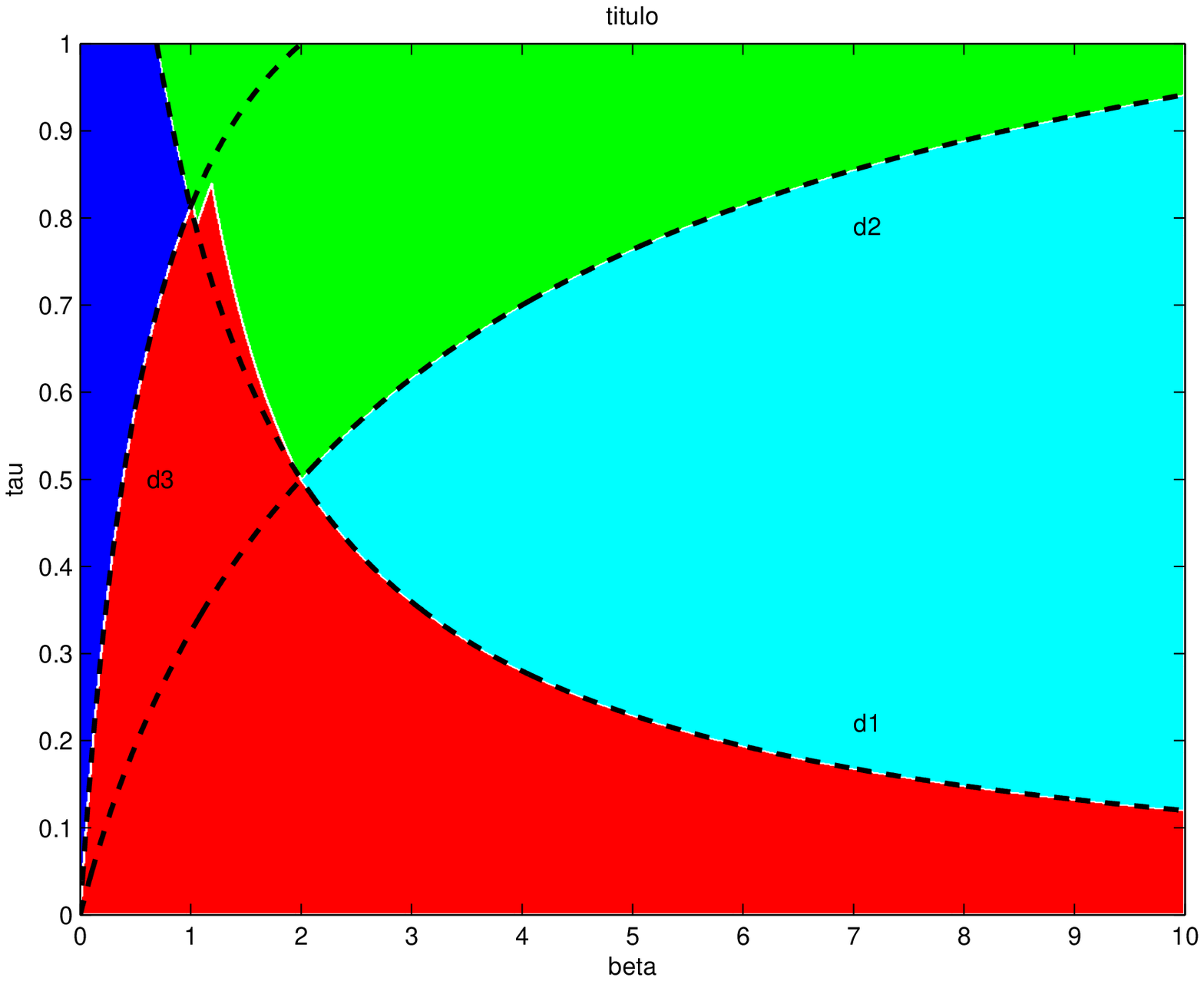}
\caption{Star graph, WAL setting, $n=5$.}
\label{fig:star_a5}
\end{figure}
\begin{figure}[h]
\centering
\psfrag{beta}[t]{$\beta$}
\psfrag{tau}{$\tau$}
\psfrag{titulo}{}
\psfrag{d1}{$\delta_1(\beta)$}
\psfrag{d2}{$\delta_2(\beta)$}
\psfrag{d3}{$\delta_3(\beta)$}
\includegraphics[width=0.65\columnwidth]{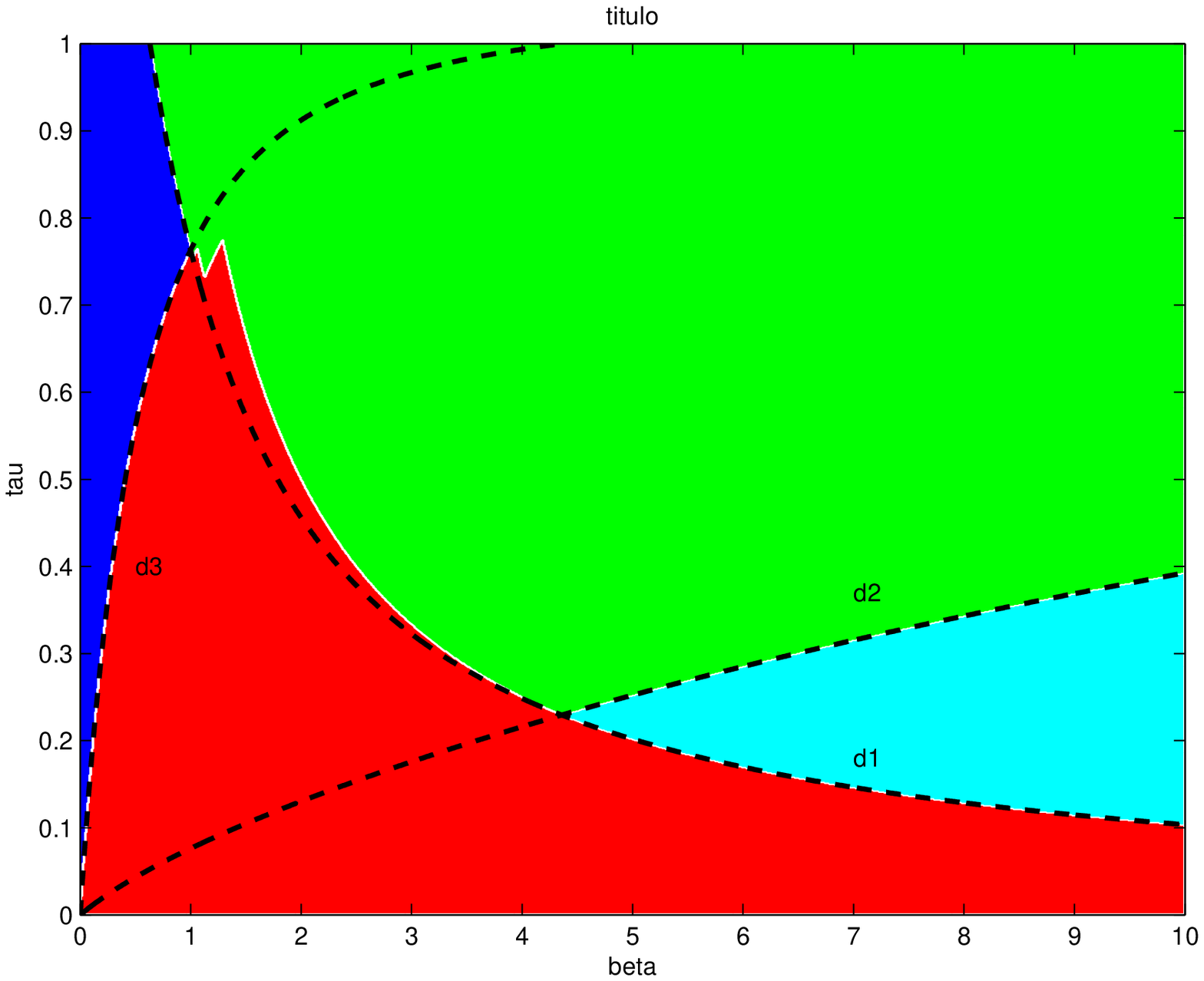}
\caption{Star graph, WAL setting, $n=20$.}
\label{fig:star_a20}
\end{figure}

\begin{proof}
By using the same argument as in Lemma \ref{lem:star}, one has $\theta_2(t)=\theta_3(t)=\dots=\theta_n(t)$, and, being $a(0)$ as in \eqref{eq:theta0}, also 
$a_2(t)=a_3(t)=\dots=a_n(t)$, $\forall t\geq 0$. Hence, in the sequel we will refer only to $\theta_2(t)$ and $a_2(t)$.
\\
i) Since $p(1)=Ga(0)$, one has
\begin{align*}
    p_1(1) &= \frac{\beta}{\beta+n-1},  \vspace*{1mm}\\
    p_2(1) &=  \frac{1}{\beta+1} .
\end{align*}
Therefore $a(1)=\1$ if and only if the following conditions are satisfied
\begin{align}
    \frac{\beta}{\beta+n-1} & \geq \theta_1(1) = \frac{n-1}{\beta+n-1} \,\tau, \label{eq:c1s}\vspace*{1mm} \\
     \frac{1}{\beta+1} & \geq \theta_2(1) =  \frac{\beta}{\beta+1} \, \tau, \label{eq:c2s}
\end{align}
which are equivalent to 
$$
\tau \leq \min \left\{ \frac{1}{\beta},\frac{\beta}{n-1} \right\}.
$$
Conversely, $a(1)=0$ if and only if 
$$
\tau > \max \left\{ \frac{1}{\beta},\frac{\beta}{n-1} \right\}.
$$
Clearly, if $\frac{1}{\beta} < \tau \leq \frac{\beta}{n-1}$, one has $a(1)=a(0)$. This can occur only if $\beta > \sqrt{n-1}$. Then, according to Corollary \ref{cor:star}, one has that $\theta_1(t)$ is monotonically increasing, while $\theta_2(t)$ is decreasing. Therefore, one will have $a(t)=a(0)$ until either $p_1(t)<\theta_1(t)$ or $p_2(t)\geq \theta_2(t)$. From~\eqref{eq:cor0star} and \eqref{eq:ths}, both conditions will never occur if the following inequalities hold
\begin{align}
  \frac{\beta}{\beta+n-1}&\geq \frac{(n-1)(\beta+1)}{n\beta +2(n-1)} \, \tau, \label{eq:condivastar} \vspace*{1mm} \\
  \frac{1}{\beta+1} &\leq \frac{(n-1)(\beta+1)}{n\beta +2(n-1)} \, \tau, \label{eq:condivbstar}
\end{align}
which correspond to $\delta_1(\beta) \leq \tau \leq \delta_2(\beta)$.  If \eqref{eq:condivbstar} is violated, i.e. $\tau < \delta_1(\beta)$, one eventually has $a(\bar{t})=\1$ for some $\bar{t}$. Similarly, if  \eqref{eq:condivastar} is violated, i.e. $\tau > \delta_2(\beta)$, one will have $a(\bar{t})=0$ for some $\bar{t}$.
\\
ii) Let $1 \leq \beta \leq \sqrt{n-1}$. From the discussion in item i), it remains to analyze the situation in which $\frac{\beta}{n-1} < \tau \leq \frac{1}{\beta}$.
By comparing with \eqref{eq:c1s}-\eqref{eq:c2s}, this assumption implies $a_1(1)=0$, $a_2(1)=1$ and $p(2)=Ga(1)$ so that
\begin{align*}
    p_1(2) &= \frac{n-1}{\beta+n-1},  \vspace*{1mm}\\
    p_2(2) &=  \frac{\beta}{\beta+1} .
\end{align*}
Through straightforward manipulations, it is easy to show that for all $\beta \leq \sqrt{n-1}$ it holds
\begin{equation}
\label{eq:th1at2}
\theta_1(2)=\left(
\frac{\beta(n-1)}{(\beta+n-1)^2} + \frac{\beta(n-1)}{(\beta+n-1)(\beta+1)}
\right) \tau < \frac{n-1}{\beta+n-1}
\end{equation}
which leads to $a_1(2)=1$. On the other hand, if $\theta_2(2) > \frac{\beta}{\beta+1}$, one will have $a_2(2)=0$ and hence $a(2)=a(0)$.
Being $-1 < \lambda_s \leq 0$ (see Corollary \ref{cor:star}), from \eqref{eq:th1at2} and the violation of the inequality in \eqref{eq:c1s}, one has 
$$
\frac{\beta}{\beta+n-1} < \theta_1(t) <\frac{n-1}{\beta+n-1}
$$
for all $t\geq 1$.
This means that $a(t)$ will keep oscillating between $[1,\ 0, \dots,0]'$ and $[0,\ 1, \dots,1]'$, until one of the following conditions is violated
\begin{align}
\theta_2(t) &> \frac{\beta}{\beta+1} ~~~~\mbox{for even $t$}; \label{eq:th2even} \vspace*{1mm}\\
\theta_2(t) &\leq \frac{1}{\beta+1} ~~~~\mbox{for odd $t$}. \label{eq:th2odd} 
\end{align}
Since $\beta \geq1$ and $\theta_2(t)$ converges, it is apparent that both conditions cannot hold indefinitely: therefore, $a(t)$ will eventually be equal either to $\1$ or to $0$, depending on which condition is violated first. By using \eqref{eq:ths2}, it is easy to show that \eqref{eq:th2even} is violated for $t=2k$ if 
$$
k \geq \left\lceil
\log_{\lambda_s^2} \left( \frac{\beta}{\tau} \frac{r+1}{\beta+1} - r \right)
\right\rceil
$$
while \eqref{eq:th2odd} is violated for $t=2k+1$ if
$$
k \geq \left\lceil
\log_{\lambda_s^2} \frac{1}{\lambda_s} \left( \frac{1}{\tau} \frac{r+1}{\beta+1} - r \right)
\right\rceil
$$
which leads to condition \eqref{eq:fractal}.
\\
iii) If $\beta<1$, by following the same reasoning as in the previous item, one has that  \eqref{eq:th2even} and  \eqref{eq:th2odd} can hold simultaneously for all $t$, provided that
\begin{equation}
\label{eq:osc}
\frac{\beta}{\beta+1} < \lim_{t\rightarrow +\infty} \theta_2(t) \leq \frac{1}{\beta+1}.
\end{equation}
By using \eqref{eq:cor0star}, this corresponds to
$$
\delta_3(\beta) < \tau \leq \delta_1(\beta).
$$

Conversely, if the leftmost inequality in \eqref{eq:osc} is violated one has $\ai=\1$, while violation of the rightmost inequality in \eqref{eq:osc} leads to $\ai=0$. 
\end{proof}
\begin{figure}[tb]
\centering
\psfrag{beta}[t]{$\beta$}
\psfrag{tau}{$\tau$}
\psfrag{titulo}{}
\includegraphics[width=0.65\columnwidth]{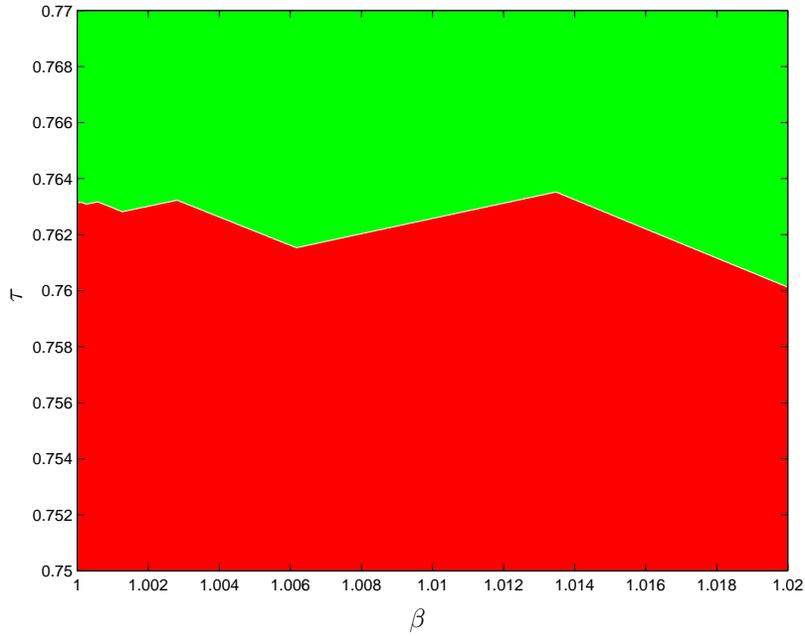}
\caption{A detail of the boundary defined by condition \eqref{eq:fractal}, for $n=20$.}
\label{fig:fractal}
\end{figure}

The next Corollary, which stems directly from the proof of Theorem \ref{th:star_a}, singles out the cases in which the asymptotic behavior is achieved in one step.
\begin{corollary}
System  \eqref{eq:modelm1},\eqref{eq:modelm3}, with $G=F$ given by \eqref{eq:Fstar} and initial condition \eqref{eq:theta0}, satisfies
\begin{itemize}
\item[-] $a(t)=\1$ for all $t \geq 1$, if and only if  $\tau \leq \min \left\{ \frac{1}{\beta},\frac{\beta}{n-1} \right\}$;
\item[-] $a(t)=0$ for all $t \geq 1$, if  and only if $\tau > \max \left\{ \frac{1}{\beta},\frac{\beta}{n-1} \right\}$.
\end{itemize}
\end{corollary}

Figures \ref{fig:star_a5}-\ref{fig:star_a20} show the asymptotic behaviors achieved in scenario WAL for different values of  $\beta$ (relative confidence parameter) and $\tau$ (initial threshold of the ordinary agents), in the cases of $5$ and $20$ agents, respectively. 
The different regions correspond to: $\ai=\1$ (red); $\ai=0$ (green); $a(t)=a(0)$, $\forall t$ (light blue); $a(t)$ switching between $[1~0~\dots~0]'$ and $[0~1~\dots~1]'$ (blue). 
The dashed curves represent the functions $\delta_i(\beta)$ defined in \eqref{eq:delta1}-\eqref{eq:delta3}. Notice that $\delta_1(\beta)$ and $\delta_2(\beta)$ intersect at $\beta=\sqrt{n-1}$, while  $\delta_1(\beta)$ and $\delta_3(\beta)$ intersect at $\beta=1$: these are the values that distinguish the three different asymptotic scenarios described by Theorem \ref{th:star_a}.

It is worth remarking the particular structure of the boundary defined by \eqref{eq:fractal}, which separates the regions in which $\ai=\1$ and $\ai=0$, when $1 \leq \beta \leq \sqrt{n-1}$.
It can be observed that this curve changes its slope an infinite number of times in any interval $\beta\in(1,1+\epsilon)$, with $\epsilon$ arbitrarily small. A detail of this behavior is shown in Figure \ref{fig:fractal}, for $n=20$.

\subsection{Uniform activity level}

In the UAL setting with the star graph interconnection, one has $F$ as in \eqref{eq:Fstar} and $G$ given by  \eqref{eq:Gstar}.
Let $r$ and $\lambda_s$ be given by \eqref{eq:ths}-\eqref{eq:lams}, and define the function of $\beta$:
\begin{align}
\mu(\beta) &=  \frac{n\beta+2(n-1)}{2(n-1)(\beta+1)}. \label{eq:mu}
\end{align}
Before proving the main result, let us introduce the following technical lemma.
\begin{lemma}
\label{lem:starb}
Let $n \geq 2$ and  $0 < \beta < \sqrt{n-1}$. Then,
$$
\frac{\beta+n-1}{n(n-1)} < \frac{\beta+1}{2\beta}.
$$
\end{lemma}

\begin{proof}
One has
$$
\frac{\beta+n-1}{n(n-1)} < \frac{\sqrt{n-1}+n-1}{n(n-1)} = \frac{1+\sqrt{n-1}}{n\sqrt{n-1}} \leq \frac{1+\sqrt{n-1}}{2\sqrt{n-1}} < \frac{\beta+1}{2\beta}
$$
where the latter inequality comes from the fact that $ \frac{\beta+1}{2\beta}$ is a strictly decreasing function of $\beta$ in the interval $0 < \beta < \sqrt{n-1}$.
\end{proof}

\begin{theorem}
\label{th:star_b}
System  \eqref{eq:modelm1},\eqref{eq:modelm3}, with $F$ and $G$ given by \eqref{eq:Fstar} and \eqref{eq:Gstar}, respectively, and initial condition \eqref{eq:theta0}, exhibits the following asymptotic behaviors.
\begin{itemize}
\item[i)] For $\beta \geq \sqrt{n-1}$,
if
\begin{equation}
\label{eq:fractal_b}
\left\lceil
\log_{\lambda_s} \left( 1-  \frac{r+1}{r \, n\,  \tau } \right)
\right\rceil
\geq
\left\lceil
\log_{\lambda_s} \left( \frac{1+r}{2 \tau}  - r \right)
\right\rceil
\end{equation}
then $\ai=\1$. Otherwise, $\ai = 0$.
\item[ii)] For $\beta < \sqrt{n-1}$,
\begin{itemize}
\item[-] if $\tau < \mu(\beta)$, then $\ai = \1$; \vspace*{2mm}
\item[-] if $\tau > \mu(\beta)$, then $\ai = 0$; \vspace*{2mm}
\item[-] if $\tau = \mu(\beta)$, then $a(t)$ oscillates indefinitely between $[1,\ 0, \dots,0]'$ and $[0,\ 1, \dots,1]'$.
\end{itemize}
\end{itemize}
\end{theorem} 

\begin{figure}[tb]
\centering
\psfrag{beta}[t]{$\beta$}
\psfrag{tau}{$\tau$}
\psfrag{d1}{$\mu(\beta)$}
\psfrag{eq}[c]{\eqref{eq:fractal_b}}
\psfrag{titulo}{}
\includegraphics[width=0.65\columnwidth]{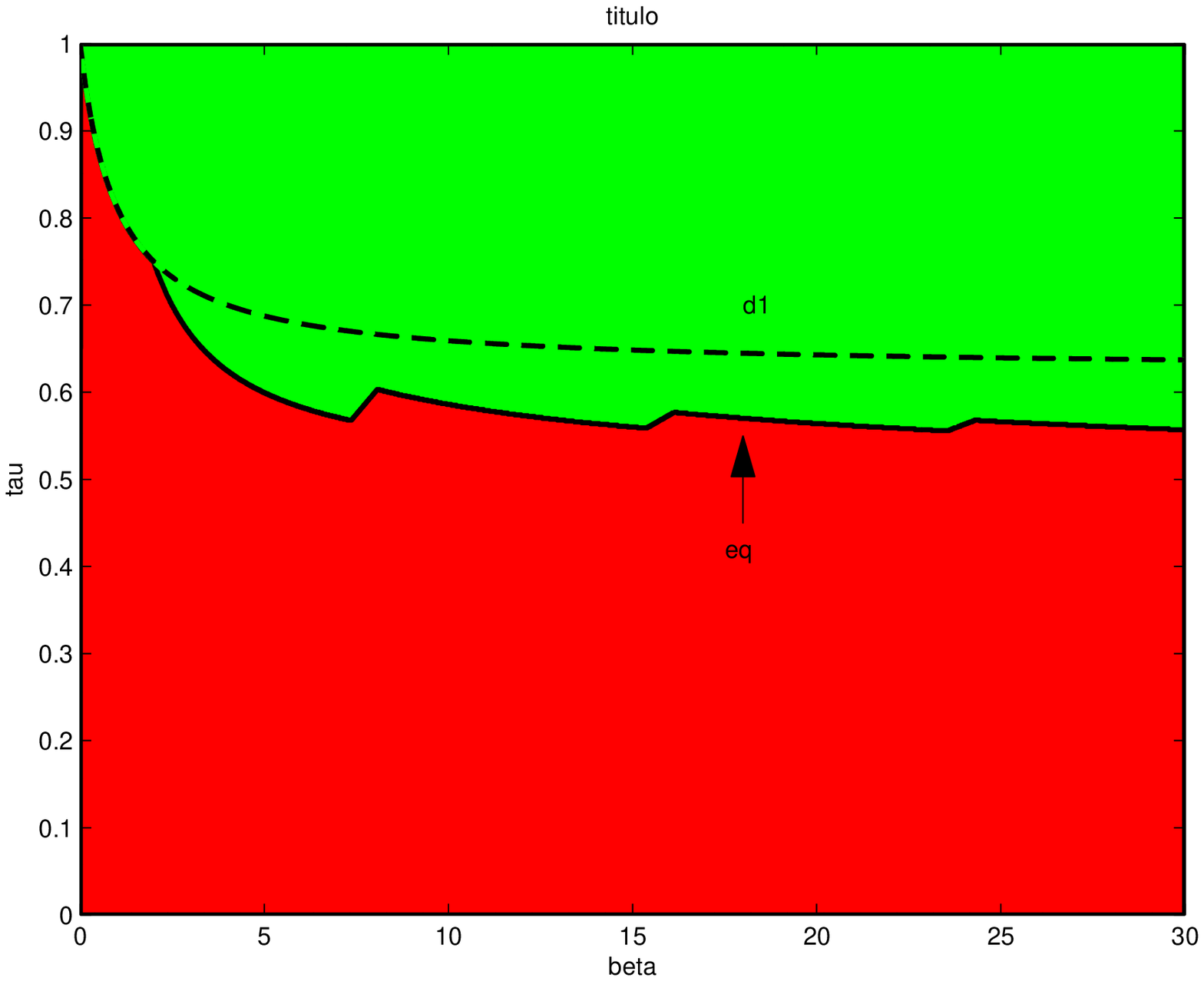}
\caption{Star graph, UAL setting, $n=5$.}
\label{fig:star_b5}
\end{figure}
\begin{figure}[tb]
\centering
\psfrag{beta}[t]{$\beta$}
\psfrag{tau}{$\tau$}
\psfrag{titulo}{}
\psfrag{d1}{$\mu(\beta)$}
\psfrag{eq}[c]{\eqref{eq:fractal_b}}
\includegraphics[width=0.65\columnwidth]{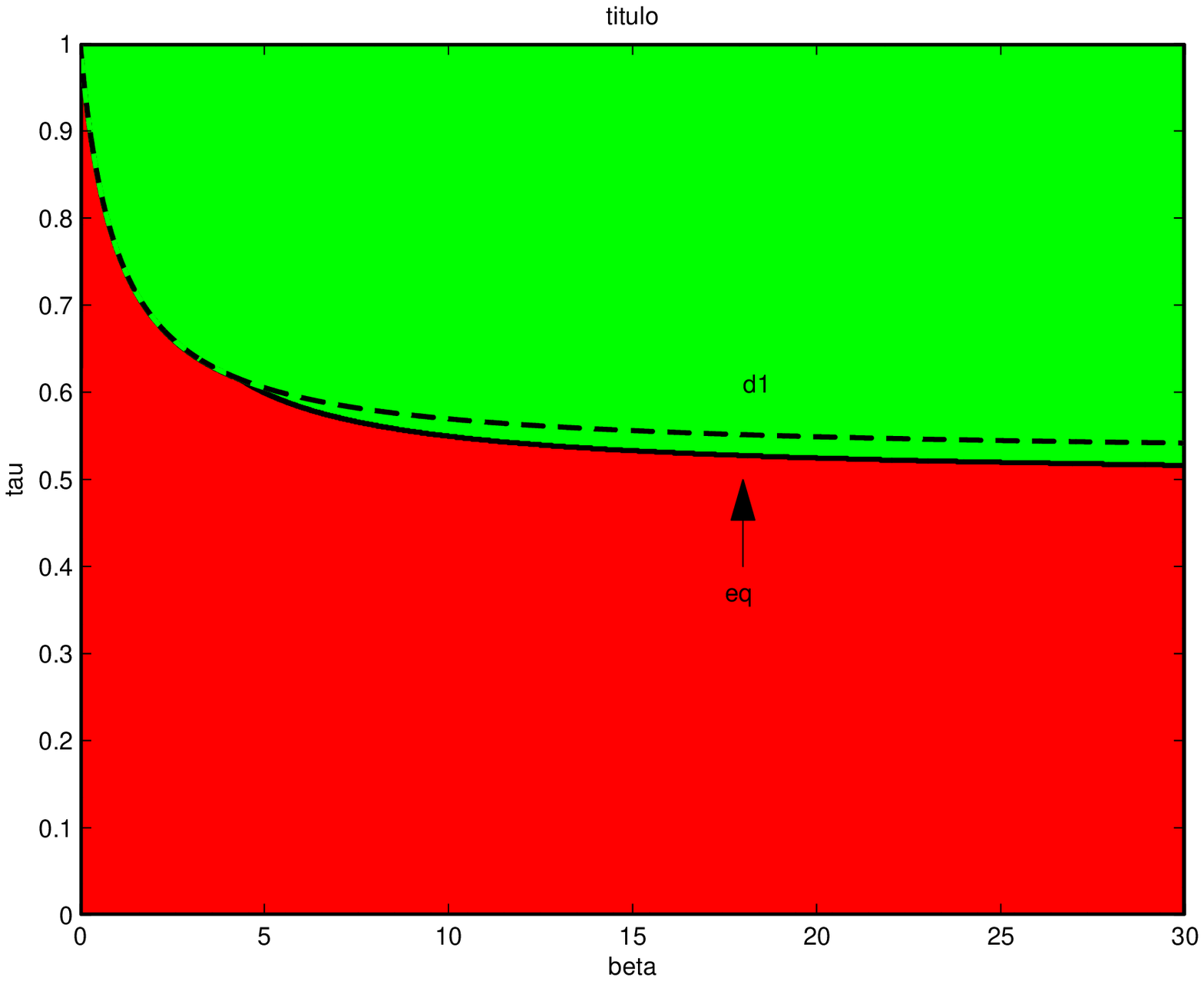}
\caption{Star graph, UAL setting, $n=20$.}
\label{fig:star_b20}
\end{figure}

\begin{proof}
As in Theorem \ref{th:star_a}, it holds $\theta_2(t)=\theta_3(t)=\dots=\theta_n(t)$ and $a_2(t)=a_3(t)=\dots=a_n(t)$, $\forall t\geq 0$. Hence, we can refer only to $\theta_2(t)$ and $a_2(t)$.
\\
i) Being $p(1)=Ga(0)$, one has $p_1(1)=\frac{1}{n}$ and $p_2(1)=\frac{1}{2}$ .
Therefore $a(1)=\1$ if and only if the following conditions are satisfied
\begin{align*}
    \frac{1}{n} & \geq \theta_1(1) = \frac{n-1}{\beta+n-1} \,\tau, \vspace*{1mm} \\
     \frac{1}{2} & \geq \theta_2(1) =  \frac{\beta}{\beta+1} \, \tau, 
\end{align*}
which are equivalent to 
$$
\tau \leq \min \left\{ \frac{\beta+n-1}{n(n-1)},\frac{\beta+1}{2\beta} \right\}.
$$
Conversely, $a(1)=0$ if and only if 
$$
\tau > \max \left\{ \frac{\beta+n-1}{n(n-1)},\frac{\beta+1}{2\beta} \right\}.
$$
If $\frac{\beta+1}{2\beta} < \tau \leq \frac{\beta+n-1}{n(n-1)}$, one has $a(1)=a(0)$. 
Notice that, due to Lemma \ref{lem:starb}, such a $\tau$ exists only if $\beta \geq \sqrt{n-1}$. Hence, according to Corollary \ref{cor:star}, the thresholds  $\theta_1(t)$  and $\theta_2(t)$ are, respectively, monotonically increasing and decreasing. Therefore, one will have $a(t)=a(0)$ until either $p_1(t)<\theta_1(t)$ or $p_2(t)\geq \theta_2(t)$, which, according to \eqref{eq:ths1}-\eqref{eq:ths2}, correspond to
\begin{align}
  \frac{1}{n}&< \displaystyle{\frac{r}{1+r} \tau \left(1- \lambda_s^t \right),}  \label{eq:condivastarb} \vspace*{1mm} \\
  \frac{1}{2} &\geq  \displaystyle{\frac{1}{1+r} \tau \left(r+ \lambda_s^t \right),}  \label{eq:condivbstarb}
\end{align}
where $r$ and $\lambda_s$ are given by \eqref{eq:ths} and \eqref{eq:lams}.
Through straightforward manipulations, \eqref{eq:condivastarb}-\eqref{eq:condivbstarb} lead respectively to
\begin{align*}
  t&\geq t_0 \triangleq \displaystyle{ \left\lceil
\log_{\lambda_s} \left( 1-  \frac{r+1}{r \, n\,  \tau } \right)
\right\rceil}, \vspace*{2mm} \\
  t&\geq t_\1 \triangleq  \displaystyle{\left\lceil
\log_{\lambda_s} \left( \frac{1+r}{2 \tau}  - r \right)
\right\rceil}.  
\end{align*}
Clearly, if $t_0 > t_\1$ one has $\ai=\1$, while if $t_0 < t_\1$, $\ai=0$. In case that $t_0 = t_\1$, one gets $a(t_0)=a(t_\1)=[0~1~1~\dots~1]'$, which leads to $p_1(t_0+1)=\frac{n-1}{n}$ and $p_2(t_0+1)=\frac{1}{2}$. Since $\theta_2(t)$ is decreasing, $\theta_2(t_0+1)<\theta_2(1)\leq \frac{1}{2}$. On the other hand $\theta_1(t)$ is increasing and hence 
\begin{equation}
\label{eq:asyth1}
\theta_1(t_0+1)<\lim_{t\rightarrow +\infty} \theta_1(t) = \frac{r}{1+r} \tau = \frac{(n-1)(\beta+1)}{n\beta+2(n-1)} \tau \leq \frac{n-1}{n}
\end{equation}
where it is easy to show that the latter inequality holds for all $n \geq 2$, being $\tau \leq 1$. Hence $a(t_0+1)=\1$ and $\ai=\1$. This proves condition \eqref{eq:fractal_b}.
\\
ii) Let $\beta < \sqrt{n-1}$. Then, due to Lemma \ref{lem:starb},
$\frac{\beta+n-1}{n(n-1)} < \frac{\beta+1}{2\beta}$. Therefore, following the reasoning in item i), $a(1)$ cannot be equal to $a(0)$. In particular, $a(t)$ will keep oscillating between $[1,\ 0, \dots,0]'$ and $[0,\ 1, \dots,1]'$ if the following conditions are satisfied
\begin{align}
&\theta_1(t) > \frac{1}{n}~,~~\theta_2(t) \leq \frac{1}{2}~ ~~~~\mbox{for odd $t$}; \label{eq:th2oddb} \vspace*{1mm}\\
&\theta_1(t) \leq \frac{n-1}{n}~,~~\theta_2(t) > \frac{1}{2}~ ~~~~\mbox{for even $t$}. \label{eq:th2evenb} 
\end{align}
Corollary \ref{cor:star} states that both $\theta_1(t)$ and $\theta_2(t)$ converge asymptotically to $\frac{r}{1+r} \tau$. Moreover, due to \eqref{eq:ths1}, $\theta_1(t)$ is increasing for all even $t$. Thanks to \eqref{eq:asyth1}, one can conclude that the first condition in \eqref{eq:th2evenb} will hold indefinitely. On the other hand, the second condition in \eqref{eq:th2oddb} will be eventually violated if and only if 
$$
\frac{r}{1+r} \tau > \frac{1}{2}
$$
which corresponds to $\tau > \mu(\beta)$. This leads to $\ai=0$.
Conversely, if $\tau < \mu(\beta)$, either the first condition in \eqref{eq:th2oddb} or the second condition in \eqref{eq:th2evenb} will eventually be violated, thus leading to $\ai=\1$. Finally, when $\tau=\mu(\beta)$, all conditions \eqref{eq:th2oddb}-\eqref{eq:th2evenb} will hold indefinitely, thus leading to oscillations of $a(t)$ between $[1,\ 0, \dots,0]'$ and $[0,\ 1, \dots,1]'$.
\end{proof}

\begin{corollary}
System  \eqref{eq:modelm1},\eqref{eq:modelm3}, with $F$ and $G$ given by \eqref{eq:Fstar} and \eqref{eq:Gstar}, respectively, and initial condition \eqref{eq:theta0}, satisfies
\begin{itemize}
\item[-] $a(t)=\1$ for all $t \geq 1$, if and only if  $\tau \leq \min \left\{ \frac{\beta+n-1}{n(n-1)},\frac{\beta+1}{2\beta} \right\}$;
\item[-] $a(t)=0$ for all $t \geq 1$, if  and only if $\tau > \max  \left\{ \frac{\beta+n-1}{n(n-1)},\frac{\beta+1}{2\beta} \right\}$.
\end{itemize}
\end{corollary}

Figures \ref{fig:star_b5}-\ref{fig:star_b20} show the asymptotic behaviors achieved in scenario UAL for different values of  $\beta$ and $\tau$, for $n=5$ and $n=20$, respectively. 
Colors have the same meaning as in Figures \ref{fig:star_a5}-\ref{fig:star_a20}.
The dashed line represents the function $\mu(\beta)$ defined in \eqref{eq:mu}, while the solid line corresponds to condition \eqref{eq:fractal_b} . 

As for the complete graph, also for the star graph it can be observed that the WAL scenario shows a wider variety of asymptotic behaviors with respect to the UAL scenario.
In particular, in the latter the initial activity pattern $a(0)$ is never maintained indefinitely and the persistent oscillations for $\beta<1$ occur only under the singular condition $\tau=\mu(\beta)$ (while in the WAL scenario they show up for the entire range of $\tau$ values). Notice that these oscillations are due to the shyness of the agents, which are less confident in their own opinion than in that of their neighbors, thus leading to persistent switchings between activity and inactivity. 
Moreover, it can be shown that in the UAL scenario, one has $\ai=\1$ whenever $\tau\leq\frac{1}{2}$, irrespectively of the number of agents $n$ and of the confidence parameter $\beta$. Conversely, in the WAL scenario, the region in which $\ai=\1$ tends to shrink as either $n$ or $\beta$ grow.

\section{Ring graph}
\label{sec:ring}

In this section we analyze the asymptotic behavior of system \eqref{eq:modelm1}-\eqref{eq:modelm3} when the graph has a ring structure.
In the ring graph, agent $i$ has as neighbors agents $i-1$ and $i+1$, with the convention $n+1=1$ (see Figure \ref{fig:ring2}). 
\begin{figure}[tb]
\centering
\psfrag{1}{$1$}
\psfrag{2}{$2$}
\psfrag{3}{$3$}
\psfrag{4}{$4$}
\psfrag{...}{$\dots$}
\psfrag{n}{$n$}
\includegraphics[width=0.4\columnwidth]{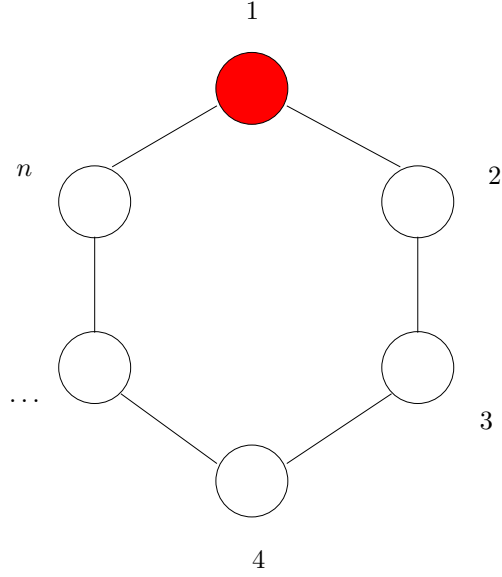}
\caption{Ring graph.}
\label{fig:ring2}
\end{figure}
We still assume that agent 1 is a radical while the others are ordinary, and we analyze the asymptotic behavior of the system starting from the initial condition \eqref{eq:theta0}.
The matrix $F$ is now given by
\begin{equation}
\label{eq:Fring}
F=
%\left(
%\begin{array}{cccccc}
\begin{bmatrix}
\frac{\beta}{\beta+2} & \frac{1}{\beta+2} & 0 & \dots & 0 & \frac{1}{\beta+2} \\
\frac{1}{\beta+2} & \frac{\beta}{\beta+2}  & \frac{1}{\beta+2} & 0 & \dots & 0 \\
0 & \frac{1}{\beta+2} & \ddots & \ddots  & \dots  & \vdots \\
\vdots & & \ddots & \ddots & \ddots & \vdots \\
0 & & & \ddots & \ddots  & \frac{1}{\beta+2} \\
\frac{1}{\beta+2} & 0 & \dots & \dots & \frac{1}{\beta+2} & \frac{\beta}{\beta+2}  
\end{bmatrix}
%\end{array}
%\right),
\end{equation}
while matrix $G$ in scenario UAL is
\begin{equation}
\label{eq:Gring}
G=
%\left(
%\begin{array}{cccccc}
\begin{bmatrix}
\frac{1}{3} & \frac{1}{3} & 0 & \dots & 0 & \frac{1}{3} \\
\frac{1}{3} & \frac{1}{3}  & \frac{1}{3} & 0 & \dots & 0 \\
0 & \frac{1}{3} & \ddots & \ddots  & \dots  & \vdots \\
\vdots & & \ddots & \ddots & \ddots & \vdots \\
0 & & & \ddots & \ddots  & \frac{1}{3} \\
\frac{1}{3} & 0 & \dots & \dots & \frac{1}{3} & \frac{1}{3} 
\end{bmatrix}
%\end{array}
%\right).
\end{equation}

In order to simplify the treatment, hereafter only the case in which $n$ is odd (i.e., the number of ordinary agents is even) will be considered. Moreover, let $h=\frac {n-1}{2}$.
The following technical result, relying on the properties of circulant matrices \cite{D79}, provides the analytic expression for the threshold evolution according to equation \eqref{eq:modelm1}.
\begin{lemma}
\label{lem:thresholdsring}
Consider the dynamic model~\eqref{eq:modelm1}, with matrix $F$ given by \eqref{eq:Fring}, $\theta(0) = [0,\ \tau, \dots,\tau]'$, $0<\tau < 1$ and odd $n$. Then, 
\begin{equation}
\label{eq:thring}
\theta(t) = \frac{n-1}n \tau \1 - \tau   \sum_{k=1}^{h} (\lambda_k)^t v_k,
\end{equation}
where
\begin{align}
\lambda_k &= \frac 1 {\beta+2} \left(\beta+2\cos\left(k \frac {2 \pi}n\right) \right), \label{eq:lamkring}\\
v_k &= \frac 2 n \left[1\ \cos\left(k \frac {2 \pi} n\right)\ \dots \  \cos\left((n-1) k \frac {2 \pi} n\right) \right]', \label{eq:vkring}
\end{align}
for $k=1, \dots, h$, and  $h=\frac {n-1}{2}$.
\end{lemma}

\begin{proof}
Matrix $F$ in \eqref{eq:Fring} is a circulant matrix, i.e. it has the form
$$
F=\left(
\begin{array}{ccccc}
f_0 & f_1 & f_2 & \dots  & f_{n-1} \\
f_{n-1} &  f_0 & f_1 &  \dots  & f_{n-2} \\
\vdots & \vdots & \vdots & \ddots & \vdots \\
f_1 & f_2 & f_3 & \dots  & f_0 \\
\end{array}
\right),
$$
where $f_0 = \frac \beta {\beta +2}$, $f_1=f_{n-1}= \frac 1 {\beta+2}$ and $f_i= 0$, for $i \neq 0,1,n-1$. The eignevalues and eigenvectors of a circulant matrix can be computed analytically (e.g., see \cite{D79}). Let $\omega_k = e^{k \frac {2\pi} n j}$, where $j=\sqrt{-1}$. The eigenvalues of $F$ are given by
$$
\lambda_k = \sum_{i=0}^{n-1} f_i \omega_k^{i} =  \frac 1 {\beta +2} \left(\beta + 2\cos\left(k \frac {2\pi} n \right) \right), \quad k=0,\dots,n-1.
$$
The eigenvectors of $F$ are given by
\begin{equation}
\label{eq:uk}
u_k = \frac 1 {\sqrt{n}} [1\ \omega_k\ \dots\ \omega_k^{n-1}]', \quad k=0,\dots,n-1. 
\end{equation}
and form an orthonormal basis. A circulant matrix can always be diagonalized. Let $U = [u_0\ u_1\ \dots\ u_{n-1}]$, then
$
F = U \Lambda U^*,
$
where $\Lambda = \text{diag}(\lambda_0,\dots,\lambda_{n-1})$ and $U^*$ is the conjugate transpose of $U$ (Th. 3.2.1 in \cite{D79}). Observing that $\lambda_0 = 1$, $\lambda_k = \lambda_{n-k}$, $k=1,\dots,h$, and $u_0= \frac 1 {\sqrt{n}} \1$, the evolution of the thresholds $\theta(t)$ can be written as
\begin{align*}
\theta(t) &= U \Lambda^t U^* \theta(0) = \frac {n-1} n \tau \1 + \sum_{k=1}^{h} \lambda_k^t (u_k u_k^* + u_{n-k} u_{n-k}^*) \theta(0) \\
& = \frac {n-1} n \tau \1  - \frac {\tau} {\sqrt{n}} \sum_{k=1}^{h} \lambda_k^t (u_k + u_{n-k}),
\end{align*}
where the last equality comes from $u_k^*\theta(0)=-\frac {\tau} {\sqrt{n}}$, since $u_k^*\1=0$, $k=1,\dots, n-1$. The thesis \eqref{eq:thring} easily follows by noting that the entries of $u_k$ in \eqref{eq:uk} are such that $\omega_k^l + \omega_{n-k}^l = 2 \cos \left (l \frac {2 \pi} n \right)$, for $l=0,\dots,n-1$ and $k=1,\dots,h$.
\end{proof}

From \eqref{eq:thring} and \eqref{eq:vkring} it can be checked that $\theta_{i+1}(t)=\theta_{n-i+1}(t)$, for $i=1,\dots,h$ and for all $t$. Hence, due to the structure of matrices $F$ and $G$ resulting from the ring interconnection, one will have also $a_{i+1}(t)=a_{n-i+1}(t)$ and $p_{i+1}(t)=p_{n-i+1}(t)$, for $i=1,\dots,h$. The following lemma gives other useful properties that will be instrumental to proving the main result. 
\begin{lemma}
\label{lem:thresholdsring2}
Consider the same assumptions as in Lemma \ref{lem:thresholdsring} and let $\beta \geq 1$. Then for all $t \geq 1$ the following statements hold
\begin{itemize}
\item[i)] $\theta_i(t) \leq \theta_{i+1}(t)$, for $i=1,\dots,h$;
\item[ii)] $\theta_i(t-1) \leq \theta_{i+1}(t)$, for $i=1,\dots,h$.
\end{itemize}
\end{lemma}

\begin{proof}
i) By \eqref{eq:theta0}, the statement is true at $t=0$. Now, let the statement hold at time $t$. Then, recalling that $\theta_{h+1}(t)=\theta_{h+2}(t)$ for all $t$, one has for $i=2,\dots,h$ 
$$
\begin{array}{rcl}
\theta_{i+1}(t+1)-\theta_i(t+1)&=&
\displaystyle{\frac{1}{\beta+2} \theta_i(t) + \frac{\beta}{\beta+2} \theta_{i+1}(t) + \frac{1}{\beta+2} \theta_{i+2}(t)}
\\
&& \displaystyle{
-\left(   
\frac{1}{\beta+2} \theta_{i-1}(t) + \frac{\beta}{\beta+2} \theta_{i}(t) + \frac{1}{\beta+2} \theta_{i+1}(t)
\right) \geq 0,}
\end{array}
$$
while, being $\beta \geq 1$, 
$$
\begin{array}{rcl}
\theta_{2}(t+1)-\theta_1(t+1)&=&
\displaystyle{
\frac{1}{\beta+2} \theta_1(t) + \frac{\beta}{\beta+2} \theta_{2}(t) + \frac{1}{\beta+2} \theta_{3}(t)
- \frac{\beta}{\beta+2} \theta_1(t)}\\
&& \displaystyle{- \frac{2}{\beta+2} \theta_{2}(t) 
~\geq~
\frac{\beta-1}{\beta+2} \left(
\theta_{2}(t) - \theta_{1}(t) 
\right) 
~\geq~ 0.}
\end{array}
$$
Therefore, the claim holds by induction.
\\
ii) 
%Being $\theta(0)$ given by \eqref{eq:theta0} and
%$$
%\theta(1) = \left( \frac{2}{\beta+2} ~~  \frac{\beta+1}{\beta+2} ~~ 1 ~~\dots ~~ 1 ~~ \frac{\beta+1}{\beta+2} \right)
%$$
%the statement is true for $t=1$.
By applying the result in item i), for all $i=1,\dots,h$, one has
$$
\begin{array}{rcl}
\theta_{i+1}(t)&=& \displaystyle{\frac{1}{\beta+2} \theta_i(t-1) + \frac{\beta}{\beta+2} \theta_{i+1} (t-1) + \frac{1}{\beta+2} \theta_{i+2}(t-1)}
\\
&\geq& \displaystyle{ \left( \frac{1}{\beta+2} + \frac{\beta}{\beta+2} + \frac{1}{\beta+2} \right) \theta_i(t-1) ~=~ \theta_i(t-1).}
\end{array}
$$
\end{proof}

\subsection{Weighted activity level}

Now, we are ready to characterize the asymptotic behavior in the WAL scenario, for the ring graph interconnection. In order to streamline the presentation, only the case $\beta \geq 1$ will be treated.

\begin{theorem}
\label{th:ring_a}
Let $\mathcal{E}$ be a ring graph, $F=G$ given by \eqref{eq:Fring} as in the  WAL setting, $n$ odd and $\beta \geq 1$.
Define
\begin{eqnarray}
q_j(\beta)&=&\inf_t \left\{
\frac{n-1}{n}-\frac{2}{n} \sum_{k=1}^{h} (\lambda_k)^t \cos\left( j k \frac{2\pi}{n} \right) \right\}
~,~~~j=1,\dots,h,
\label{eq:qibeta}
\\
q_0(\beta) &=& \frac{n-1}{n\,\beta}
\end{eqnarray}
where 
%$h=\frac{n-1}{2}$ and
%$$
%\lambda_k = 1 - \frac 2 {\beta+2} \left(1-\cos\left(k \frac {2 \pi}{n} \right) \right).
%$$
$\lambda_k$ is given by \eqref{eq:lamkring}, and set $\hbar=\lfloor \frac{h}{2} \rfloor$.
Then, system  \eqref{eq:modelm1},\eqref{eq:modelm3}, with initial condition \eqref{eq:theta0}, exhibits the following asymptotic behaviors:
\begin{itemize}
\item[i)] If $\tau \leq \frac{n}{n-1} \frac{1}{(\beta+2)}$, then $\ai=\1$.
\item[ii)]
If 
\begin{equation}
\label{eq:ringcondalfai}
\frac{1}{(\beta+2)q_{j+1}(\beta)} \leq \tau \leq \frac{1}{(\beta+2)q_{j}(\beta)},
\end{equation}
for some $j \in \{0,1,\dots,\hbar\}$, then 
$\ai=\alpha_j$, given by 
\begin{equation}
\label{eq:alphai}
\begin{array}{rcccl}
\alpha_j = \big[ & \hspace*{-3mm}
\underbrace{ ~1 ~ \dots ~ 1 ~} & 0 ~ \dots ~  0  &  \underbrace{ ~1 ~ \dots  ~ 1 ~} &\hspace*{-3mm} \big]',~~j~=1,\dots,\hbar.
\\
& \hspace*{-3mm} j+1 && j &
\end{array}
\end{equation}
\item[iii)] 
If $ \tau > \max \left\{ \frac{1}{(\beta+2)q_{0}(\beta)} \, , \, \frac{1}{(\beta+2)q_{1}(\beta)} \right\} $, then $\ai=0$.
\end{itemize}
\end{theorem}
\begin{figure}[tb]
\centering
\psfrag{beta}[t]{$\beta$}
\psfrag{tau}{$\tau$}
\psfrag{q0}{$\frac{1}{(\beta+2)q_0(\beta)}$}
\psfrag{q1}{$\frac{1}{(\beta+2)q_1(\beta)}$}
\psfrag{q2}{$\frac{5}{4(\beta+2)}$}
\psfrag{titulo}{}
\includegraphics[width=0.65\columnwidth]{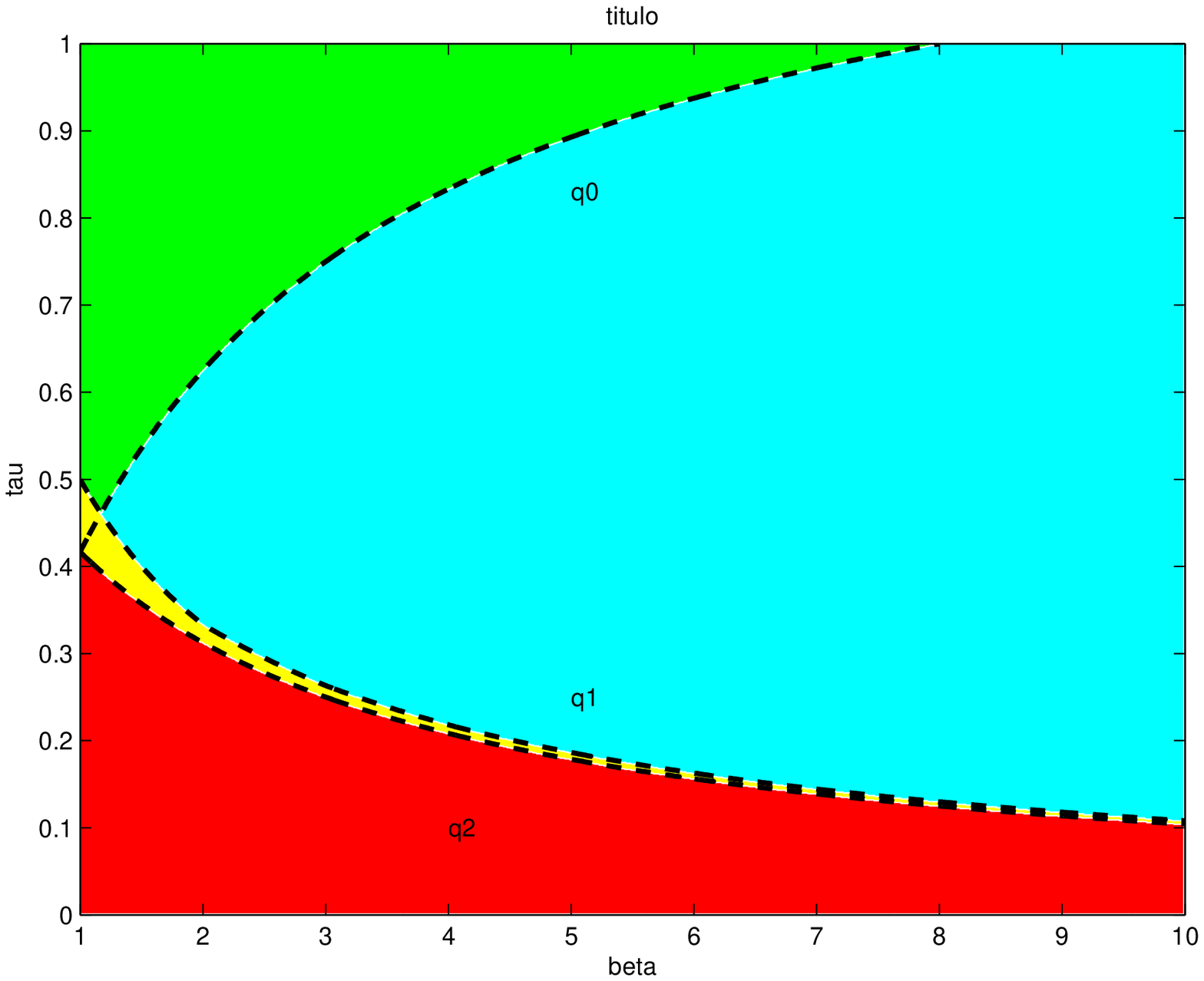}
\caption{Ring graph, WAL setting, $n=5$.}
\label{fig:ring_a5}
\end{figure}
\begin{figure}[tb]
\centering
\psfrag{beta}[t]{$\beta$}
\psfrag{tau}{$\tau$}
\psfrag{titulo}{}
\psfrag{q0}{$\frac{1}{(\beta+2)q_0(\beta)}$}
\psfrag{q1}{$\frac{1}{(\beta+2)q_1(\beta)}$}
\psfrag{qh}{$\frac{21}{20(\beta+2)}$}
\includegraphics[width=0.65\columnwidth]{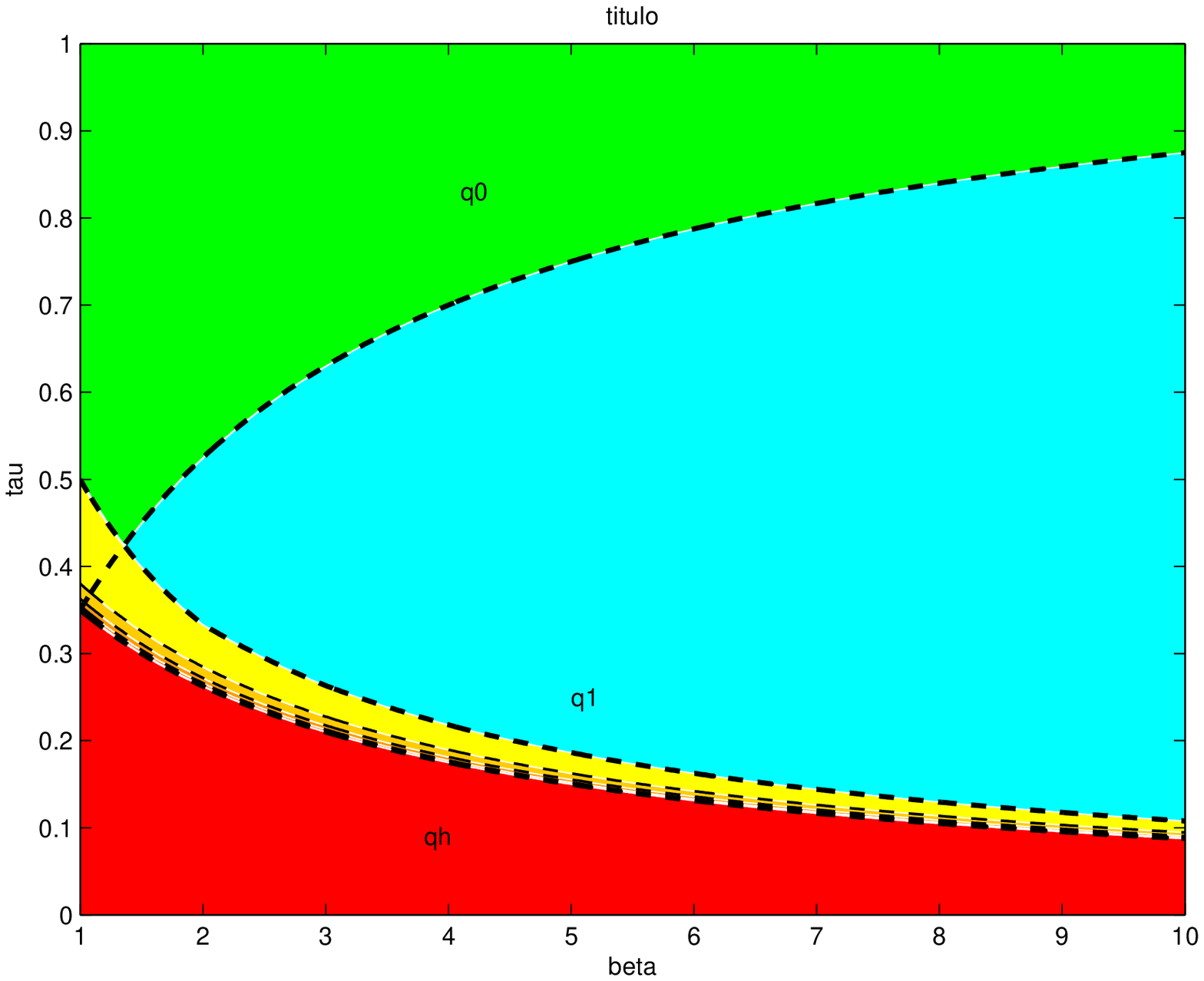}
\caption{Ring graph, WAL setting, $n=21$.}
\label{fig:ring_a21}
\end{figure}

\begin{proof}
Before proving the single items, let us show by induction that $p_i(t) \geq p_{i+1}(t)$, for $i=1,\dots,h$ and for all $t\geq0$.
Being 
$$
p(0)= \left[ \frac{\beta}{\beta+2} ~~\frac{1}{\beta+2}~~0~~\dots ~~ 0 ~~ \frac{1}{\beta+2} \right]'
$$
and $\beta \geq 1$, the statement is true at $t=0$. Let $p_i(t) \geq p_{i+1}(t)$, $i=1,\dots,h$, at time $t$. Then, being $\theta_i(t) \leq \theta_{i+1}(t)$ according to item i) in Lemma \ref{lem:thresholdsring2}, one has that $a(t)=\alpha_j$ in \eqref{eq:alphai} for some $j$, which in turn gives 
\begin{equation}
\begin{array}{rcccccccl}
p(t+1)=G\alpha_j=\big[ & \hspace*{-3mm}
\underbrace{ ~1 ~ \dots ~ 1 ~} & \frac{\beta+1}{\beta+2} & \frac{1}{\beta+2} & \underbrace{0 ~ \dots ~  0}  &  \frac{1}{\beta+2} & \frac{\beta+1}{\beta+2} & \underbrace{ ~1 ~ \dots  ~ 1 ~} &\hspace*{-3mm} \big]'
\\
& \hspace*{-3mm} j &&& n-2j-3 &&& j-1 &
\end{array}
\label{eq:peqGalpha}
\end{equation}
and hence $p_i(t+1) \geq p_{i+1}(t+1)$, $i=1,\dots,h$. 
\\
i) According to the previous discussion, $a(t)$ is either equal to $\1$ or $0$, or it takes values $\alpha_j$ in \eqref{eq:alphai}, which then lead to $p(t+1)$ as in \eqref{eq:peqGalpha}.
Being $\lim_{\rightarrow + \infty} \theta(t)= \frac{n-1}{n} \tau$, if $\frac{n-1}{n} \tau < \frac{1}{\beta+2}$, one has that eventually $a(t)$ will switch from $\alpha_{j-1}$ to $\alpha_{j}$, for all $j=1,\dots,h$, thus leading to $\ai=\alpha_{h}=\1$.
\\
ii) Let $j \in \{1,2,\dots,\hbar\}$. From \eqref{eq:thring} and \eqref{eq:qibeta}, one gets
$$
\inf_t \theta_{j+1}(t) = \tau \, q_j(\beta).
$$
Through some tedious manipulations, it is possible to show that the infimum in \eqref{eq:qibeta} is equal to $\frac{n-1}{n}$ for $j=\hbar+1,\dots,h$, and it is approached asymptotically as $t$ grows to infinity. Conversely, for $j=1,\dots,\hbar$, the infimum satisfies $q_j(\beta) < \frac{n-1}{n}$ and is attained at some finite $t_j^*$. Moreover, being $\theta_{j+1}(t)=\tau$ for all $t \leq j-1$, one obviously has that $t_j^* \geq j$.
By following the same reasoning as in item i), the switch from $\alpha_j$ to $\alpha_{j+1}$ will never occur if $ \tau \, q_j(\beta) > \frac{1}{\beta+2}$. Now, let \eqref{eq:ringcondalfai} hold.
From the rightmost inequality, by applying item ii) in Lemma \ref{lem:thresholdsring2}, one gets
$$
\frac{1}{\beta+2} \geq \theta_{j+1}(t_j^*) \geq \theta_{j}(t_j^*-1) \geq \dots  \geq \theta_{2}(t_j^*-j+1).
$$
Being $t_j^* \geq j$, one has $t_j^*-j+1\geq 1$. Hence, all the thresholds $\theta_i$, $i=2,\dots,j$ will take a value below $\frac{1}{\beta+2}$ in successive times, thus guaranteeing that all the switchings from $\alpha_{i-1}$ to $\alpha_i$ will eventually occur. This proves that condition \eqref{eq:ringcondalfai} leads to $\ai=\alpha_j$.
\\
iii) If $ \tau > \frac{1}{(\beta+2)q_1(\beta)}$, one has that the switching from $a(0)=\alpha_0$ to $\alpha_1$ never occurs. In such a case, one gets
$$
p(t)=Ga(0)= 
\left[  \frac{\beta}{\beta+2} ~~ \frac{1}{\beta+2} ~~  0 ~~ \dots ~~  0  ~~\frac{1}{\beta+2}  \right]'
$$
and $a(t)=\alpha_0$ indefinitely, unless $\lim_{t \rightarrow +\infty} \theta_1(t) > \frac{\beta}{\beta+2}$, for some $t$. The latter condition corresponds to $\frac{n-1}{n} \tau > \frac{\beta}{\beta+2}$, i.e., $\tau >  \frac{1}{(\beta+2)q_{0}(\beta)}$. 
\end{proof}

Figures \ref{fig:ring_a5}-\ref{fig:ring_a21} show the asymptotic behaviors achieved in scenario WAL for different values of $\beta$ and $\tau$, in the cases of 5 and 21 agents, respectively.
The different regions correspond to: $\ai=\1$ (red); $\ai=0$ (green); $a(t)=a(0)$, $\forall t$ (light blue); $\ai=\alpha_j$ (different shades of yellow). In particular, for the case $n=5$ only $\ai=\alpha_1$ is present, while for $n=21$ one can observe five regions corresponding to $\ai=\alpha_j$, $j=1,\dots,5$, the lightest (and largest) one corresponding to $\alpha_1$. 
The dashed curves represent the boundaries defined by Theorem \ref{th:ring_a} as functions of $\tau$ and $\beta$.

In Figure \ref{fig:ring_a21}, it can be noticed that the regions in which $\ai=\alpha_j$ tend to shrink as $j$ grows, while they approach the region in which $\ai=\1$. Moreover, these regions also shrink as $\beta$ grows. On the other hand, their number increases with the number of agents, being proportional to $\frac{n}{4}$. Apart from these regions, the other asymptotic patterns are the same as for the star graph, but the region in which $\ai=a(0)$ is much larger in the ring network, while those in which $0$ or $\1$ are achieved are significantly reduced (compare, e.g., Figures \ref{fig:star_a20} and  \ref{fig:ring_a21}). 

\subsection{Uniform activity level}

Let us consider the UAL setting. The asymptotic behavior for $\beta \geq 1$ is described by the next result.

\begin{theorem}
\label{th:ring_b}
Let $\mathcal{E}$ be a ring graph, $F$ and $G$ given by \eqref{eq:Fring} and \eqref{eq:Gring}, respectively, $n$ odd and $\beta \geq 1$.
Let the functions
$q_j(\beta)$, $j=1,\dots,h$, be defined as in \eqref{eq:qibeta}.
Then, system  \eqref{eq:modelm1},\eqref{eq:modelm3}, with initial condition \eqref{eq:theta0}, exhibits the following asymptotic behaviors:
\begin{itemize}
\item[i)] If $\tau \leq \frac{n}{3(n-1)}$, then $\ai=\1$.
\item[ii)]
If 
$$
\frac{1}{3q_{j+1}(\beta)} \leq \tau \leq \frac{1}{3q_{j}(\beta)},
$$
for some $j \in \{1,2,\dots,\hbar\}$, then 
$\ai=\alpha_j$, given by \eqref{eq:alphai}.
\item[iii)] 
If $ \tau >  \frac{n-1}{3n}$,  then $\ai=0$.
\end{itemize}
\end{theorem}
\begin{figure}[tb]
\centering
\psfrag{beta}[t]{$\beta$}
\psfrag{tau}{$\tau$}
\psfrag{q1}{$\frac{1}{(\beta+2)q_1(\beta)}$}
\psfrag{q0}{$\frac{5}{12}$}
\psfrag{titulo}{}
\includegraphics[width=0.65\columnwidth]{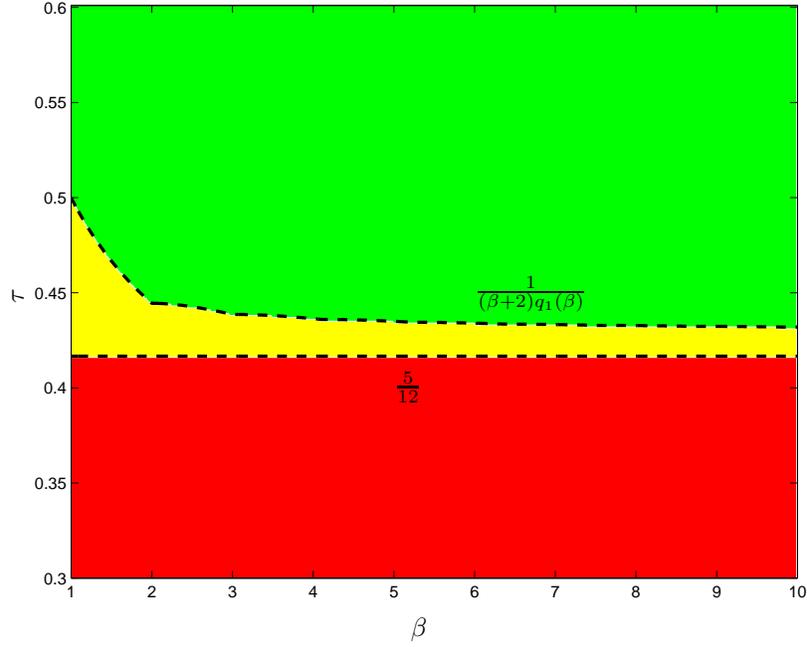}
\caption{Ring graph, UAL setting, $n=5$.}
\label{fig:ring_b5}
\end{figure}
\begin{figure}[tb]
\centering
\psfrag{beta}[t]{$\beta$}
\psfrag{tau}{$\tau$}
\psfrag{titulo}{}
\psfrag{q2}{$\frac{1}{(\beta+2)q_2(\beta)}$}
\psfrag{q1}{$\frac{1}{(\beta+2)q_1(\beta)}$}
\psfrag{qh}{$\frac{21}{60}$}
\includegraphics[width=0.65\columnwidth]{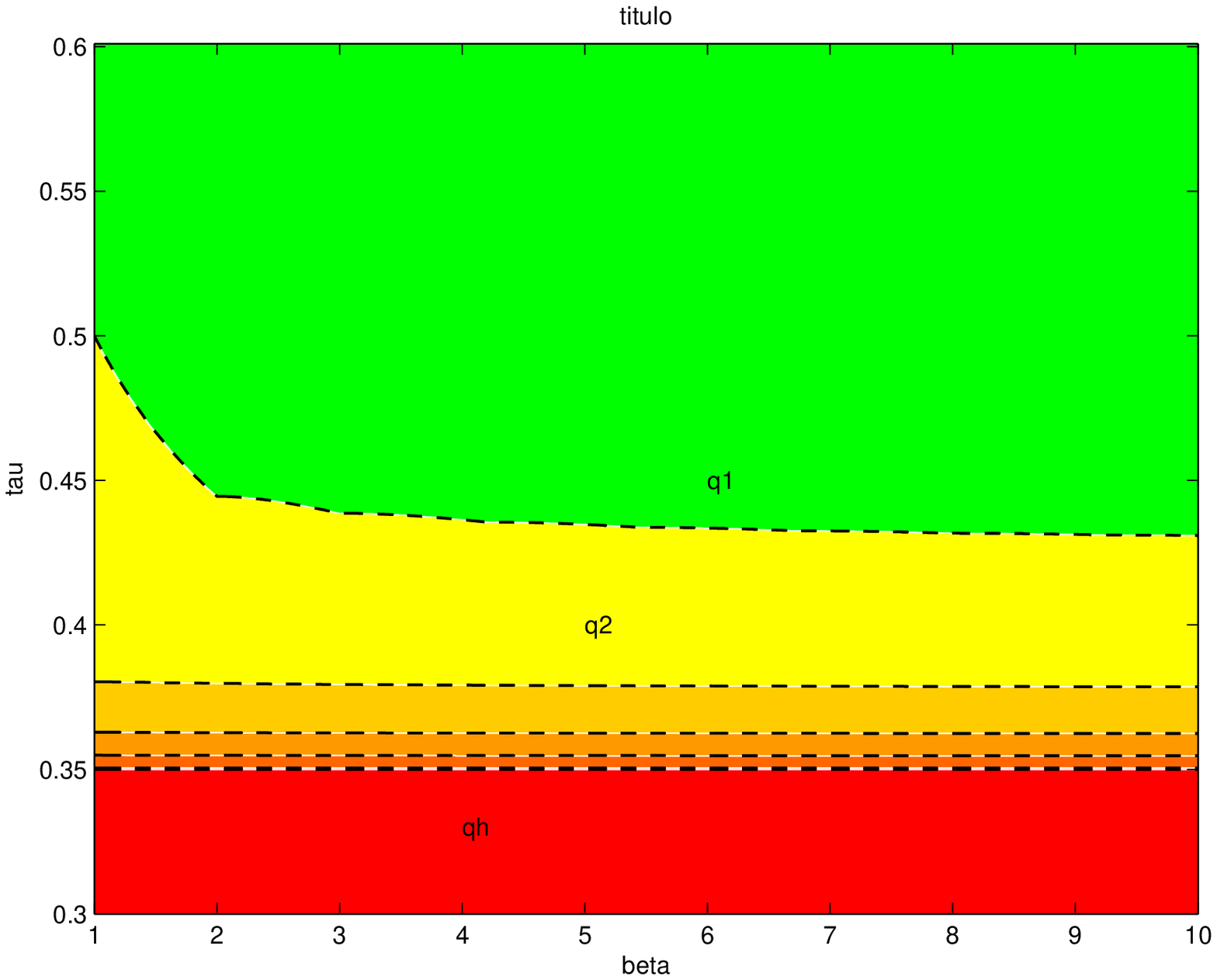}
\caption{Ring graph, UAL setting, $n=21$.}
\label{fig:ring_b21}
\end{figure}

\begin{proof}
By adopting the same argument as in the proof of Theorem \ref{th:ring_a}, it can be shown that $a(t)$ can take only the values $\1$, $0$, or $\alpha_j$ in \eqref{eq:alphai}.
In the latter case, one gets
$$
\begin{array}{rcccccccl}
p(t+1)=G\alpha_j=\big[ & \hspace*{-3mm}
\underbrace{ ~1 ~ \dots ~ 1 ~} & \frac{2}{3} & \frac{1}{3} & \underbrace{0 ~ \dots ~  0}  &  \frac{1}{3} & \frac{2}{3} & \underbrace{ ~1 ~ \dots  ~ 1 ~} &\hspace*{-3mm} \big]'.
\\
& \hspace*{-3mm} j &&& n-2j-3 &&& j-1 &
\end{array}
$$
Then, item i) and item ii) for $j=1,\dots,\hbar$ can be proven by following the same reasoning as in the proof of Theorem \ref{th:ring_a}.
\\
Concerning item iii), let us first observe that if $\tau > \frac{1}{3q_1(\beta)}$, the switching from $a(0)=\alpha_0$ to $\alpha_1$ never occurs. Being 
$$
p(1)=G\alpha_0=\left[\frac{1}{3} ~~\frac{1}{3}~~0~~ \dots ~~0 ~~ \frac{1}{3}\right]'
$$
and $q_1(\beta)<\frac{n}{n-1}$, one has $\tau \frac{n-1}{n} > \frac{1}{3}$ and hence $\lim_{t \rightarrow +\infty} \theta_1(t) > \frac{1}{3}$. This means that eventually one will have $\theta_1(t) > p_1(t)$ and therefore $\ai=0$. 
\end{proof}

Figures \ref{fig:ring_b5}-\ref{fig:ring_b21} show the asymptotic behaviors achieved in scenario UAL for different values of  $\beta$ and $\tau$, for $n=5$ and $n=21$, respectively. 
Colors have the same meaning as in Figures \ref{fig:ring_a5}-\ref{fig:ring_a21}. The range of $\tau$ is reduced to highlight the presence of the regions in which $\ai=\alpha_j$.
The dashed line represents the boundaries defined by Theorem \ref{th:ring_b} as function of $\tau$ and $\beta$.

According to Theorem \ref{th:ring_b}, in the UAL scenario it never occurs that $\ai=a(0)$, i.e., the initial condition cannot be maintained indefinitely. Once again, this is a major difference with respect to what happens in the WAL setting (compare Figures \ref{fig:ring_b5}-\ref{fig:ring_b21} with Figures \ref{fig:ring_a5}-\ref{fig:ring_a21}). Moreover, a larger value of the self-confidence parameter $\beta$ enlarges the gap between the extreme asymptotic behaviors $\ai=\1$ and $\ai=0$ in the WAL scenario, while this is not the case in the UAL case. 

By comparing the results obtained in the two considered scenarios, it can be concluded that the choice of the matrix $G$ plays a key role in defining the pattern of asymptotic behaviors in all the interconnection topologies considered in the paper.

\begin{figure}[ht]
  \centering
  \includegraphics[width=0.5\columnwidth]{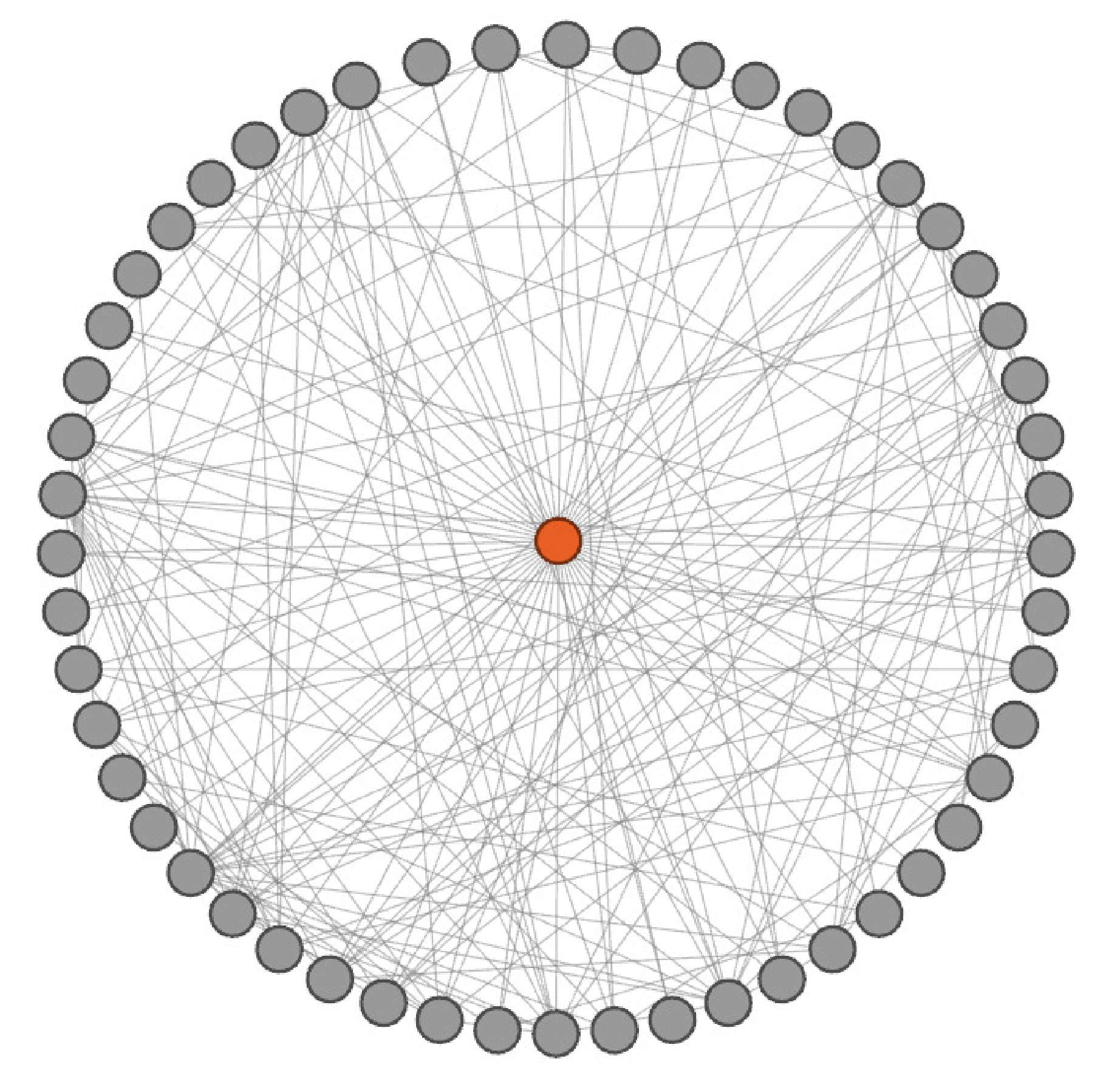}
  \caption{\aggiungi{Ego network used for simulation, including $n=53$ nodes and 198 edges (the red node is the ego node).}}
  \label{fig:ego}
\end{figure}

\begin{figure}[h!]
  \centering
  \psfrag{x}[t]{$\beta$}
  \psfrag{y}{$\tau$}
  \psfrag{t}{}
  \psfrag{d1}{\color{black}{\small $\delta_1(\beta)$}}
  \psfrag{d2}{\color{black}{\small $\delta_2(\beta)$}}
  \psfrag{a1}{\color{black}{\small $a_\infty = 0$}}
  \psfrag{a2}{\color{black}{\small $a_\infty = a(0)$}}
  \psfrag{a3}{\color{black}{\small $a_\infty = \1$}}
  \includegraphics[width=0.75\columnwidth]{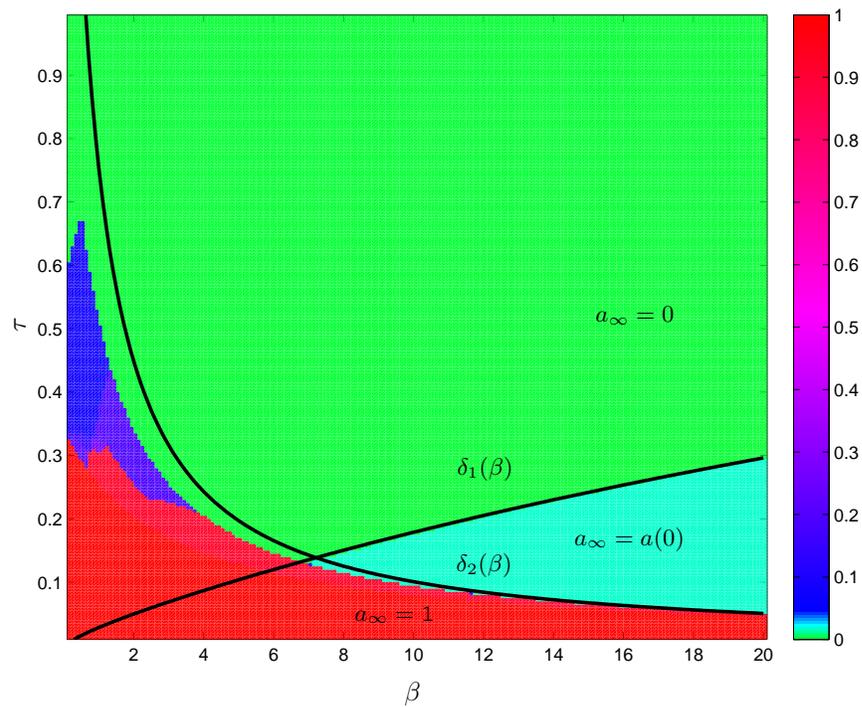}
  \caption{\aggiungi{Final fraction of radical agents when the only initial radical agent is placed at the ego node.}}
  \label{fig:egoonly}
\end{figure}

\begin{figure}[ht]
  \centering
  \begin{tabular}{cc}
    \psfrag{x}[t]{$\beta$}
    \psfrag{y}{$\tau$}
    \psfrag{t}{}
    \includegraphics[width=0.49\columnwidth]{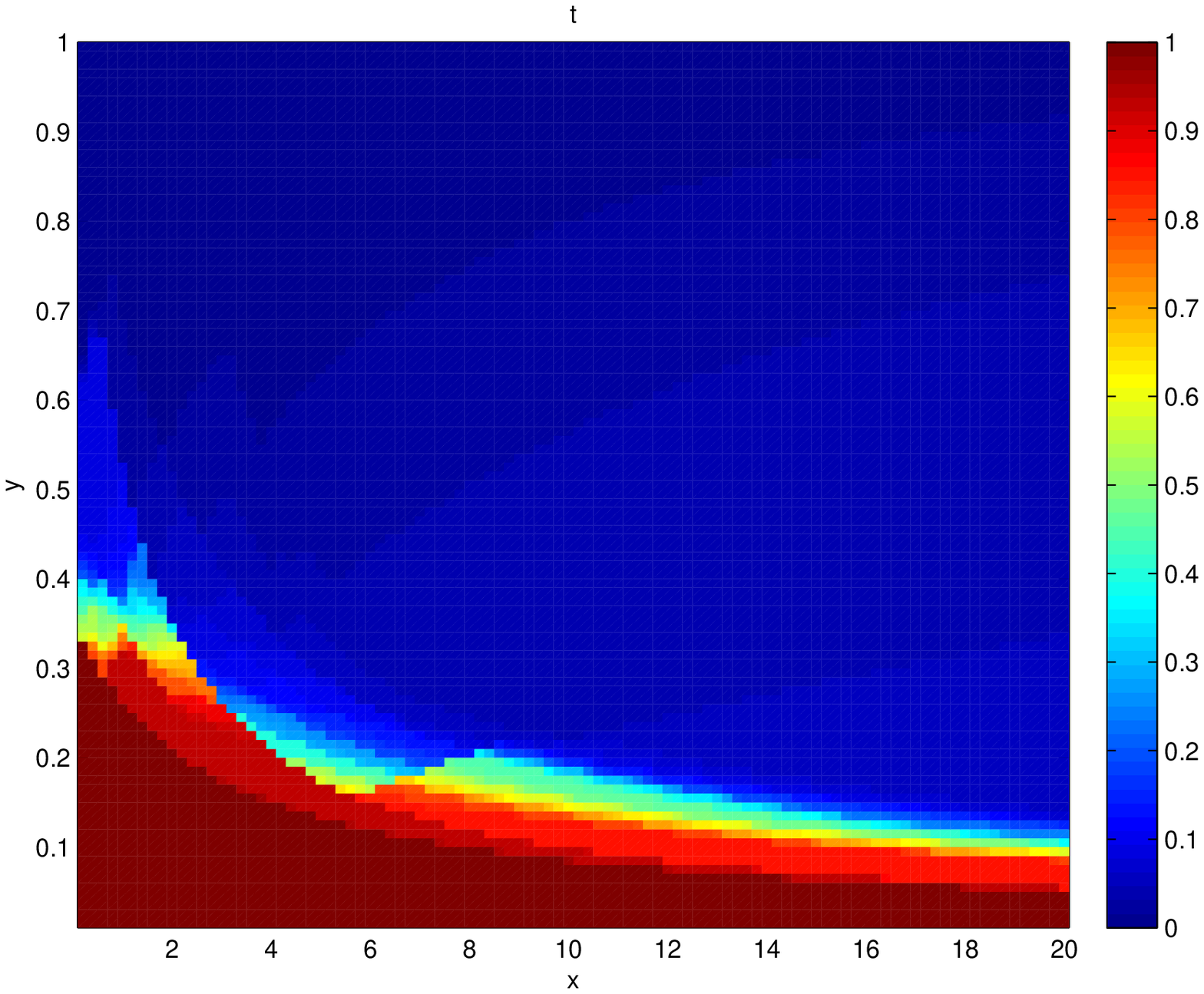} &
    \psfrag{x}[t]{$\beta$}
    \psfrag{y}{$\tau$}
    \psfrag{t}{}
    \includegraphics[width=0.49\columnwidth]{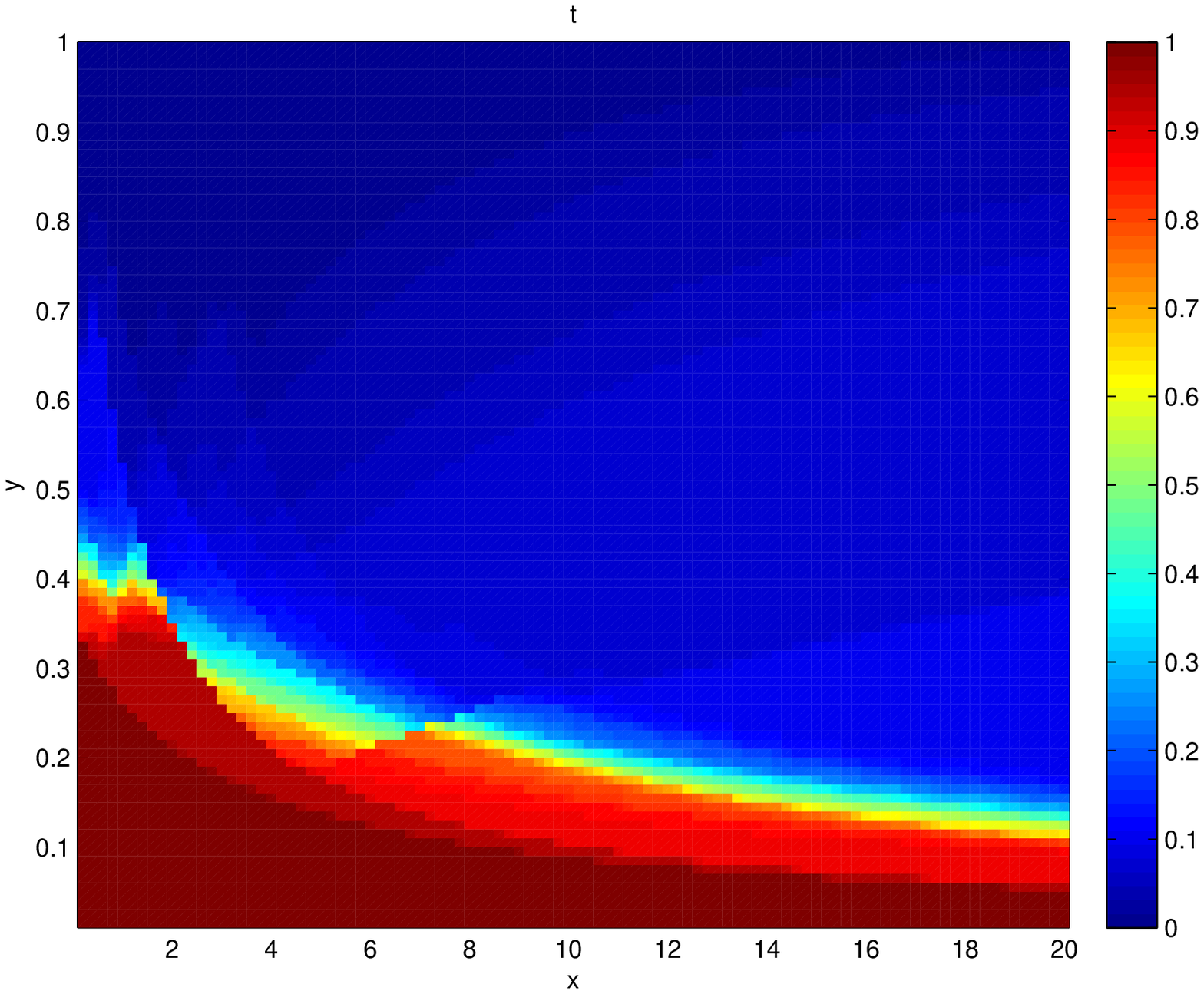} \\
    (a) $\xi =0.05$ & (b) $\xi =0.10$ \\
    \psfrag{x}[t]{$\beta$}
    \psfrag{y}{$\tau$}
    \psfrag{t}{}
    \includegraphics[width=0.49\columnwidth]{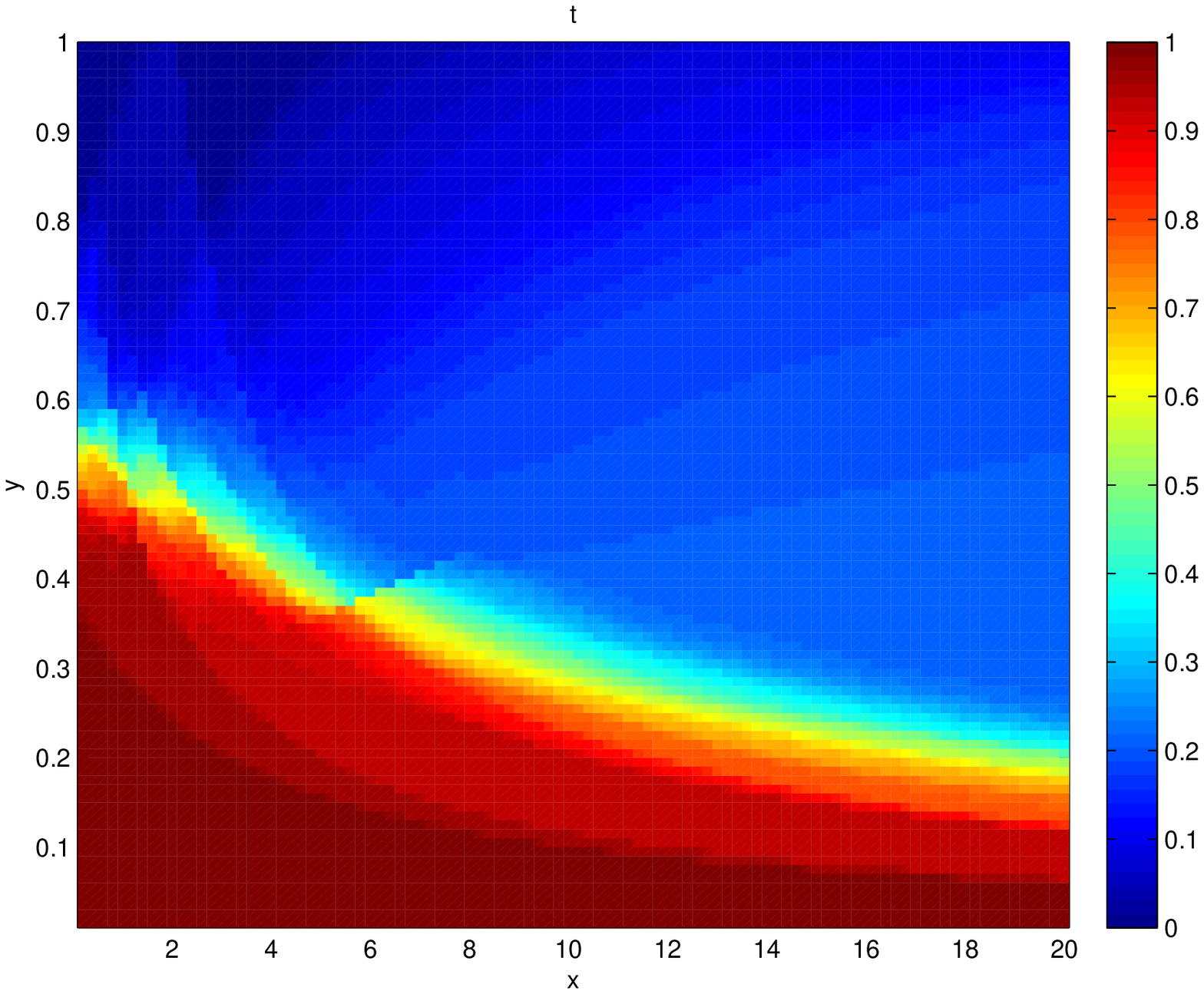} &
    \psfrag{x}[t]{$\beta$}
    \psfrag{y}{$\tau$}
    \psfrag{t}{}
    \includegraphics[width=0.49\columnwidth]{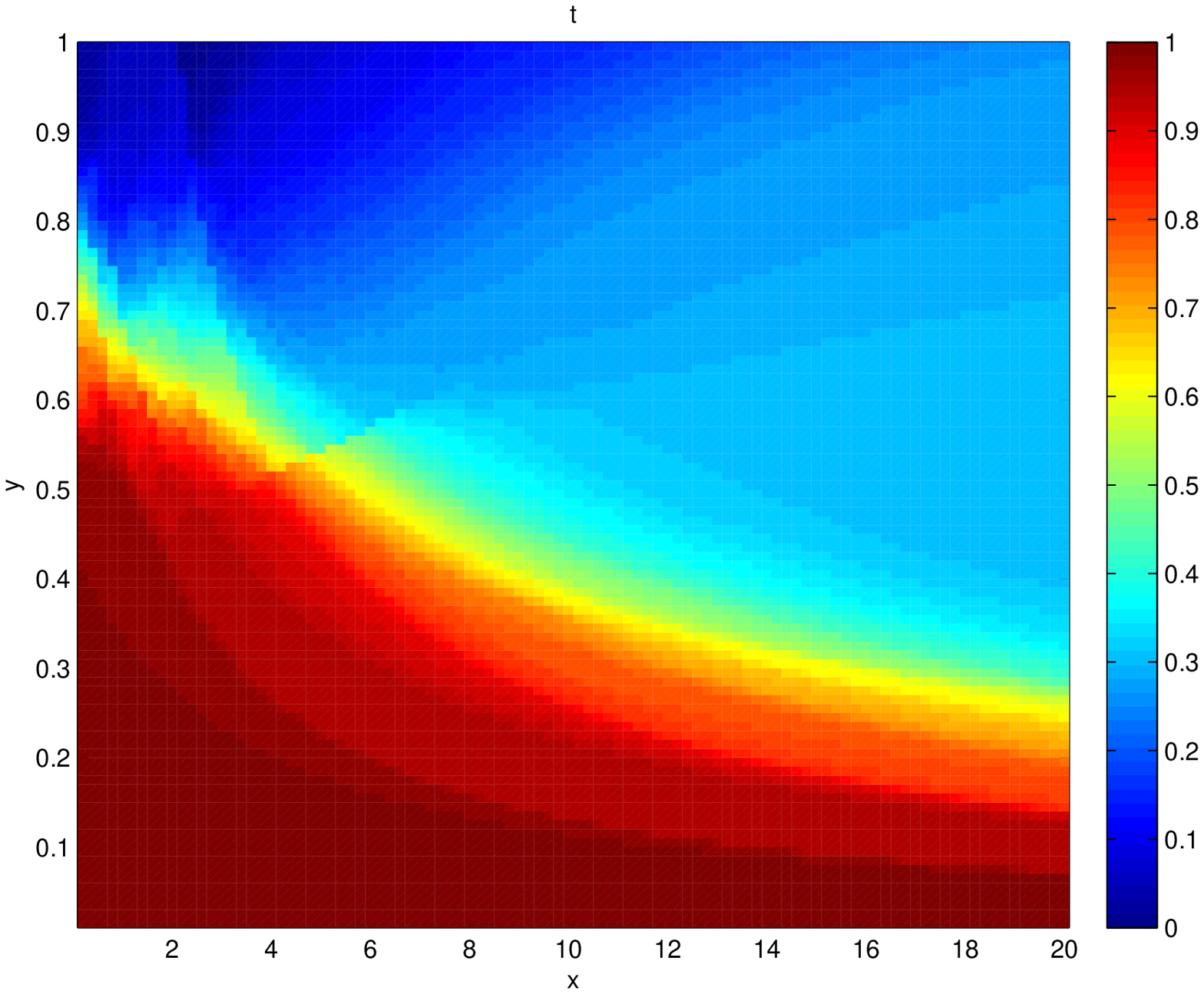} \\
    (c) $\xi =0.20$ & (d) $\xi =0.30$
  \end{tabular}
  \caption{\aggiungi{Final fraction of radical agents, averaged over 100 different initial conditions, for different initial fractions $\xi$ of radical agents.}}
  \label{fig:egores}
\end{figure}

\section{Experimental results with ego networks}
\label{sec:exp}

\aggiungi{In this section, we evaluate to what extent the analytic results presented so far, and obtained in the case of simple graph topologies, apply to more realistic social networks. To this end, we have carried out extensive simulations on a real ego network and observed how the asymptotic behaviors depend on the model parameters.}

\aggiungi{
The test network used in this study is extracted from the data set described in \cite{LM12}, consisting of ten ego networks taken from Facebook. The network selected has $n = 53$ nodes, including the ego node, and 198 edges, resulting in an average degree of 7.47 edges per node (see Fig.~\ref{fig:ego}). The highest degree node is clearly the ego node, which, by definition, is connected to all the other nodes. The node with the second highest degree has 19 incident edges, whereas 10 nodes are connected to the ego node only. Dynamic system~\eqref{eq:modelm1},\eqref{eq:modelm3} has been simulated with such an ego network constituting the underlying communication infrastructure. The WAL setting has been considered, i.e. $F=G$ is assumed in throughout this section, with the entries of $F$ given by~\eqref{eq:weights}. Different combinations of the initial threshold $\tau$ and the self-confidence weight $\beta$ have been simulated. Parameter $\tau$ ranges from 0.01 to 0.99, whereas $\beta$ varies from 0.1 to 20.
}

\aggiungi{
In the first scenario considered, the ego node is the only agent initially active. The fraction number of final active agents, as a function of $\tau$ and $\beta$, is shown in Fig.~\ref{fig:egoonly}. Among the graphs studied in Sec.~\ref{sec:complete}-\ref{sec:ring}, the network topology more similar to the ego network under consideration is clearly the star network. Although a non negligible number of additional edges between the non-central nodes are now present, the asymptotic behaviors of the system are very similar to those analytically derived in Sec.~\ref{th:star_a} (e.g., compare Fig.~\ref{fig:star_a20} to Fig.~\ref{fig:egoonly}). For large values of the self-confidence weight, i.e. for $\beta > \sqrt{n-1} \simeq 7.21$, the asymptotic behavior of the test network is in very good agreement with that predicted by Theorem~\ref{th:star_a} (point (i)), with $a_\infty$ switching from $\1$ to $a(0)$ to $0$ as $\tau$ crosses the functions $\delta_1(\beta)$ and $\delta_2(\beta)$ (see Fig.~\ref{fig:egoonly}). %\footnote{\aggiungi{Notice that the regions marked as $a_\infty=0$ and $a_\infty=a(0)$, corresponding to zero and one final active agent, respectively, have different (although very similar) colors in Fig.~\ref{fig:egoonly}.}} 
For smaller values of $\beta$ a richer variety of asymptotic behaviors are now observed. In contrast to what happens for a truly star network, in this case 13 different values of $a_\infty$ are found. Notice, however, that  three stationary asymptotic behaviors predicted by Theorem~\ref{th:star_a}, namely $a_\infty \in \{0, a(0), \1 \}$, cover more than 95\% of the simulation runs.
}

\aggiungi{
A second set of simulations has been carried out to analyze how the presence of more than one initial radical agent modifies the final distribution of active agents. To this end, the simulations have been initialized to a number of radicals equal to $\text{round}(\xi n)$, where $\xi$ denotes the fraction of initial radical agents. The identity of the initial radicals (i.e., the node where initial radicals are placed within the network) may have an impact on their ability to persuade a larger number of neighbors. Intuitively, more central or more connected nodes (in a sense, more ``popular'' agents) are able to mobilize a higher number of individuals. To mitigate such an effect, simulation results are averaged over 100 different randomly generated identities of the initially active agents. For consistency with the theoretical analysis, the ego node is always initially active. The final fraction of active agents is reported in Fig.~\ref{fig:egores}, for four different values of $\xi$ ranging from 0.05 to 0.30. It can be noticed that the smaller the number of initial radicals, the sharper the transition between regions corresponding to different asymptotic behaviors. Although somehow blurred by the averaging process, separating curves similar to those shown in Fig.~\ref{fig:fractal} can be observed. For values of $\beta$ close to one, the transition between the regions with all active and all inactive final agents is very irregular and spiked. This suggests that the fractal boundary found in Theorem~\ref{th:star_a} (point (ii)), and shown in Fig.~\ref{fig:fractal}, is revealing of a phenomenon which can be experienced also in actual ego networks.  
}

\section{Conclusions}
\label{sec:conclusions}

% From Hassanpour's "Media disruption..."
%
% Abstract: "Conventional wisdom suggests that universal lapses in media connectivity—for example, disruption of Internet and cell phone access—have a negative effect on political mobilization. On the contrary, I argue that sudden and ubiquitous interruption of mass communication can facilitate revolutionary mobilization and proliferate decentralized contention"
%
% Conclusions: "In other words, in the presence of a risk-averse majority and a radical minority, adding more links among the majority does not necessarily help mobilization."
%
% Frase da aggiungere alla conclusioni del ns. lavoro
% 

This paper has presented a class of dynamic threshold models which can be used to analyze collective actions in social networks. The main feature of this model class is that the threshold is time-varying, as it evolves according to a dynamic opinion model. This leads to the generation of complex transient dynamics, in which each agent can change her mind multiple times about undertaking the action or not, and in some cases can even lead to steady oscillating behaviors. The asymptotic activity pattern of the network clearly depends not only on the graph topology, but also on the level of self-confidence of the agents. Moreover, a crucial role is played by the selected mechanism for the computation of the neighbors' activity level, which determines how the agents decide to become active or not.

The analytic results obtained so far support the thesis proposed in \cite{H14}, and based on empirical evidence, according to which  \emph{``in the presence of a risk-averse majority and a radical minority, adding more links among the majority does not necessarily help mobilization.''} By looking at the regions corresponding to all agents becoming eventually active (e.g., red regions in Figs. \ref{fig:n20GeqF}, \ref{fig:star_a20} and \ref{fig:ring_a21}), it can be clearly seen that achieveing full mobilization in a highly connected network (e.g., a complete graph) can be harder than doing it in a less connected topology (e.g., star and ring graphs).

There are many interesting developments that can be foreseen for the proposed model class. First of all, in this study only simple graph structures have been considered. This has allowed us to derive analytic results providing a complete characterization of the asymptotic behaviors for such networks. 
\aggiungi{Despite the basic structure of the considered networks, the obtained results shed light on the potentiality of the model, and provide useful insights on the behavior of more complex structures, such as ego networks, as confirmed by numerical simulation.
}
Other studies may concern the case in which more radical agents are present in the network, or the presence of groups of ordinary agents having different initial thresholds (e.g., modeling two parties with different initial opinions about the action to be undertaken). The influence of the position of the radical agents within the network should also be investigated. Another extension of the model consists in defining groups of agents with different self-confidence levels, in order to analyze which types of dynamics arises between confident and hesitant agents. Alternative opinion dynamics models can also be considered for the threshold evolution, by adopting either time-varying or even state-dependent weights, like in the Hegselmann-Krause model \cite{BHT09}. Finally, it is worth remarking that the framework considered in this work is deterministic, but stochastic versions can be formulated. For example, the self-confidence parameter might be a random variable, thus accounting for variability in the agents' self-confidence, or the network topology itself can be stochastic, which is common in the social learning literature \cite{ADLO11}.

\bibliographystyle{IEEEtran}
\bibliography{threshold}

\end{document}